\documentclass[aps,pra,twocolumn,groupedaddress]{revtex4-2}

\usepackage[T1]{fontenc}

\usepackage{amsmath,amssymb,amsthm}
\usepackage{comment}
\usepackage{etoolbox}

\newtheorem{lemma}{Lemma}
\newtheorem{theorem}{Theorem}

\newtheorem{corollary}{Corollary}

\newtheorem{remark}{Remark}
\usepackage[colorlinks=true, allcolors=blue]{hyperref}

\usepackage{graphicx}           
\usepackage{color}
\definecolor{sorange}{rgb}{0.8, 0.3, 0}

\definecolor{warnred}{rgb}{0.8, 0.15, 0.1}

\definecolor{darkgreen}{rgb}{0.0, 0.5, 0.0}
\usepackage{mathtools}
\usepackage{orcidlink}

\begin{document}


\title{Simple Sufficient Criteria for Optimality of Entanglement Witnesses}



\author{Frederik vom Ende\,\orcidlink{0000-0002-2738-6893}}
\affiliation{Dahlem Center for Complex Quantum
Systems, Freie Universit\"at Berlin, Arnimallee 14, 14195 Berlin, Germany}

\author{Simon Cichy\,\orcidlink{0000-0002-9409-193X}}
\email{simon.cichy@fu-berlin.de}
\affiliation{Dahlem Center for Complex Quantum
Systems, Freie Universit\"at Berlin, Arnimallee 14, 14195 Berlin, Germany}


\date{\today}

\begin{abstract}
If one wants to establish optimality of a given bipartite entanglement witness, the current standard approach is to check whether it has the spanning property. Although this is not necessary for optimality, it is most often satisfied in practice, and for small enough dimensions or sufficiently structured witnesses this criterion can be checked by hand. In this work we introduce a novel characterization of the spanning property via entanglement-breaking channels, which in turn leads to a new sufficient criterion for optimality. This criterion amounts to just checking the kernel of some bipartite state. It is slightly weaker than the spanning property, but it is a lot easier to test for---by hand as well as numerically---and it applies to almost all witnesses which are known to have the spanning property. A second criterion is derived from this, where one can simply compute the expectation value of the given witness on a maximally entangled state. Finally, this approach implies new spectral constraints on witnesses as well as on positive maps.
\end{abstract}


\maketitle

\section{Introduction}
Entanglement is a core feature of quantum physics and a key resource in many quantum information tasks \cite{HHHH07,GT09}, hence detecting it is a most fundamental task. 
While deciding whether a given quantum state is entangled or separable is a notoriously hard problem \cite{Gurvits03,Gharibian10}, a powerful yet experimentally accessible criterion is to use entanglement witnesses.
Entanglement witnesses are Hermitian operators that produce non-negative expectation values on all separable states, so if the expectation value on some state is negative, then that state is guaranteed to be entangled.
Such witnesses have a close connection to Bell inequalities \cite{Terhal00,HGBL05}, and have been extended to continuous‐variable systems \cite{HE06} (where, for Gaussian states, one can actually parametrize all tangent hyperplanes to the separable covariance matrices) as well as more general infinite‐dimensional settings \cite{HG10}.  
More recently, this concept has been strengthened by the introduction of so-called ``mirrored entanglement witnesses'' \cite{BCH20,BBH23}.
Also for relations and quantitative bounds between entanglement witnesses and other entanglement measures, cf.~Ref.~\cite{EBA07}.

This concept of (linear) entanglement witnesses, quite naturally, leads to the following central question: When is a witness $W$ \textit{optimal}, in the sense that no witness can detect a larger set of entangled states than $W$ can? Not only are optimal witnesses the ``best detectors'' for entanglement (in the sense that they cannot be improved), the collection of them is also sufficient to detect all entangled states \cite{CS14}. 
Another key concept here is ``weak optimality'' which means that the hyperplane induced by a witness $W$ is tangent to the separable states (i.e., ${\rm tr}(\rho W)=0$ for some $\rho$ separable).
While this is a strictly weaker notion than optimality, it encodes the problem of deciding whether a given state is separable \cite{BHHA13}, further justifying the study of optimality and related notions. This also shows that certifying weak optimality of a witness $W$---or equivalently, computing the minimum ${\rm tr}(\rho W)$ over all separable states---is itself NP-hard;
for more on the computational complexity of the separability problem and related problems we refer to \cite{Ioannou07,Ioannou05_PhD}.

The most widely used sufficient condition for optimality of bipartite entanglement witnesses is the \textit{spanning property}: if the set of product vectors which have zero expectation value on $W$ do span the full Hilbert space, then $W$ is optimal \cite{LKCH00}. We will explore all of this in more detail in Sec.~\ref{sec_EW_pos}, but for now we stress that while the spanning property is not necessary for optimality, ``the optimality of [witnesses] without the spanning property is rather exceptional'' \cite{BSC23}.
However, verifying the spanning property in large dimensions or for unstructured witnesses quickly becomes intractable.

Another angle this can be approached from is via positive, but not completely positive maps $\Phi$ which act on an $n$-dimensional system. This is equivalent to the entanglement-witness framework via the Choi-Jamio\l{}kowski isomorphism.
If one adds trace-preservation as a constraint to $\Phi$ and then considers the few known maps of this type, one finds that many of them---e.g., the (rescaled) reduction map or the Breuer-Hall map---admit a trace (i.e., sum of eigenvalues) equal to $-n$ \cite{38054}. This is curious because it is \textit{not} explained by known spectral considerations (as they only yield $-n^2+2\leq{\rm tr}(\Phi)\leq n^2$ for any $\Phi$ positive, trace-preserving which is too weak as soon as $n\geq 3$).
On the other hand, for the positive trace-preserving maps for which the trace is equal to $-n$, the associated entanglement witness is known to be optimal.
Hence, the question: is there more to this (i.e., is this lower bound of $-n$ true and if so, does it have anything to do with optimality), or is this just a coincidence?

In this work we tackle both these problems.
As such, this paper is organized as follows: In Sec.~\ref{subsec_maps_reps} we recap some basic notions of channels and general linear maps, and Sec.~\ref{sec_EW_pos} is devoted to separability, entanglement witnesses, optimality, positive maps, and so on.
Then come our main results: In Sec.~\ref{sec_opt} we introduce sufficient criteria for optimality of bipartite entanglement witnesses. These criteria are derived from the spanning property; the core result here will be that
a witness $W$ has the spanning property if and only if it vanishes on a full-rank separable state. Again using Choi-Jamio\l{}kowski,
this translates to a new characterization of the spanning property involving entanglement-breaking channels (Thm.~\ref{thm_main_new}), so simply evaluating this on the particular entanglement-breaking channel $X\mapsto \frac{1}{\dim\mathcal H+1}(X+{\rm tr}(X){\bf1})$ leads to a sufficient criterion for optimality where one only has to check the kernel of one of two specific bipartite states (Thm.~\ref{thm:full_Schmidt_rank_to_optimal}). 
This in turn implies an even simpler criterion based on a new lower bound on the expectation value of a given witness on any maximally entangled state which, if saturated, guarantees optimality (Coro.~\ref{coro:trace_bound_from_max_entangled}).
Both new criteria of ours are a lot easier to use and test for than the spanning property---analytically as well as numerically---and yet the only witness we could find which (i)~has the spanning property but (ii)~which our criterion does \textit{not} detect as optimal is the flip in odd dimensions (Rem.~\ref{rem_kernel_weaker_than_spanning}~(i)).

Then in Sec.~\ref{sec_pos} we translate the aforementioned results into a new lower bound on the trace (entanglement fidelity) of positive maps which, if saturated, guarantees optimality of the associated witness (Thm.~\ref{thm_2}). This, in turn, leads to new bounds on the eigenvalues of witnesses (equivalently: of the Choi matrix of positive maps), cf.~Coro.~\ref{coro_wit_lambda}. These new necessary conditions for positivity of a linear map become useful, e.g., for falsifying properties like information backflow (P-divisibility) or Markovianity (CP-divisibility) \cite{Chrus22}.
Either way, throughout Sec.~\ref{sec_mainres} we explicitly show optimality of most known witnesses in a rather simple manner using our new criteria.
Finally, our conclusions as well as some follow-up questions are presented in Sec.~\ref{sec_concl}.

\section{Preliminaries \& Notation}\label{sec_prelim}

\subsection{Linear maps and their representations}\label{subsec_maps_reps}

Let us start by recalling some key concepts and establish some notation choices;
we refer to \cite{Watrous18} for more detail on anything mentioned in this section.
An operator $\rho\in\mathbb C^{n\times n}$ is positive semi-definite (denoted $\rho\geq 0$) if $\langle x|\rho|x\rangle\geq 0$ for all 
$x\in \mathbb C^n$. 
It is well-known that $\rho\geq 0$ if and only if ${\rm tr}(\rho\omega)$ for all $\omega\geq0$.
We call any positive semi-definite operator with trace $1$ a quantum state.
Also an operator is called positive definite---denoted $\rho>0$---if $\langle x|\rho|x\rangle>0$ for all $x\neq 0$. Moreover, for any bipartite operator in $\mathbb C^{m\times m}\otimes\mathbb C^{n\times n}$ 
we write ${\rm tr}_1(\cdot)\in\mathbb C^{n\times n}$, ${\rm tr}_2(\cdot)\in\mathbb C^{m\times m}$ for the usual partial trace over respective subsystem.

As for some common norms: The operator norm, i.e., the largest singular value, will be denoted $\|\cdot\|_\infty$. The trace norm, i.e., the sum of all singular values, we denote by $\|\cdot\|_1$. Moreover, the Hilbert-Schmidt inner product $\langle A,B\rangle_{\rm HS}:={\rm tr}(A^\dagger B)$ induces the Hilbert-Schmidt norm, denoted by $\|A\|_2:=\sqrt{\langle A,A\rangle_{\rm HS}}=\sqrt{{\rm tr}(A^\dagger A)}$. 

Next let us recall important classes of linear maps $\Phi:\mathbb C^{m\times m}\to\mathbb C^{n\times n}$.
First, $\Phi$ is called Hermitian-preserving if $\Phi(A)^\dagger=\Phi(A)$ for all $A\in \mathbb C^{m\times m}$ Hermitian; equivalently, $\Phi(X^\dagger)=\Phi(X)^\dagger$ for all $X\in\mathbb C^{m\times m}$.
Next, $\Phi$ is called positive if $\Phi(A)\geq 0$ for all $A\geq 0$. If ${\rm id}_k\otimes\Phi$ is positive for all $k\in\mathbb N$, then $\Phi$ is called completely positive. A completely positive map which is additionally trace-preserving is known as (quantum) channel.

Now for some convenient and well-known representations of such maps.
Any $\Phi:\mathbb C^{m\times m}\to\mathbb C^{n\times n}$ that is Hermitian-preserving can be decomposed as 
$
    \Phi = \sum_i \alpha_i K_i (\cdot) K_i^\dagger 
$
with $\alpha_i \in \mathbb R$ and $K_i \in\mathbb C^{m\times n}$ some set of orthogonal matrices under the Hilbert-Schmidt inner product, i.e., ${\rm tr} (K_i K_j^\dagger) \propto \delta_{ij}$; this is sometimes known as generalized Kraus decomposition. Moreover, $\Phi$ is completely positive if and only if $\alpha_i\geq 0$ for all $i$ in the above decomposition,
and $\Phi$ is trace-preserving if and only if $\sum_i \alpha_i K_i^\dagger K_i = \bf 1$.
Linear maps can also be represented through the Choi-Jamio\l{}kowski (henceforth CJ, for short) isomorphism
$\mathsf C(\Phi):=({\rm id}\otimes\Phi)(|\Gamma\rangle\langle\Gamma|)$, where $|\Gamma\rangle:=\sum_j|jj\rangle$ is ``the'' (unnormalized) maximally entangled state.
Its inverse maps operators $W \in \mathbb C^{m \times m} \otimes \mathbb C^{n \times n}$ to the superoperator $\mathsf C^{-1}(W) = {\rm tr}_1(((\cdot)^T \otimes {\bf 1})W)$ where, here and henceforth ${}^T$ is the usual transpose map. The use of this mapping is that it translates complete positivity into positivity of bipartite operators, i.e., $\Phi$ is completely positive if and only if $\mathsf C(\Phi)\geq 0$.
In this language the $\mathsf{FLIP}$ operator (or $\mathsf{SWAP}$ operator) $\mathbb F \in \mathbb C^{n \times m} \otimes \mathbb C^{m \times n}$, $\mathbb F = \sum_{i=1}^n\sum_{j=1}^m |i\rangle \langle j| \otimes |j \rangle \langle i|$---for the case $m=n$---is just $\mathsf C({}^T)$.
This operator has the important properties $\mathbb F(x\otimes y)=y\otimes x$ for all $x\in\mathbb C^m, y\in\mathbb C^n$, as well as $\mathbb F(A\otimes B)\mathbb F=B\otimes A$ for all $A\in\mathbb C^{m\times m}$, $B\in\mathbb C^{n\times n}$.

Another concept we need is the Choi rank of a linear map $\Phi$, which is defined as the rank of its Choi matrix $\mathsf C(\Phi)$.
If $\Phi$ is Hermitian-preserving then this is also known as Kraus rank because this number is also the smallest possible number of terms in any Kraus decomposition $= \sum_i \alpha_i K_i (\cdot) K_i^\dagger $ of $\Phi$. We say $\Phi$ has full Choi rank
if $\mathsf C(\Phi)$ is invertible (equivalently: $\ker(\mathsf C(\Phi))=\{0\}$). For completely positive $\Phi$, full Choi rank is equivalent to $\mathsf C(\Phi)$ being positive definite.

Finally, the (Hilbert-Schmidt) adjoint of a linear map $\Phi :\mathbb C^{m\times m}\to\mathbb C^{n\times n} $ is defined as the unique linear map $\Phi^\dagger :\mathbb C^{n\times n}\to\mathbb C^{m\times m} $
which for all $A,B$ satisfies
$
    \langle A, \Phi(B) \rangle_{\rm HS} = \langle \Phi^\dagger (A), B \rangle_{\rm HS} \, .
$
If $\Phi$ is Hermitian-preserving, then this is equivalent to ${\rm tr}(A\Phi(B))={\rm tr}(\Phi^\dagger(A)B)$ for all $A,B$, which is sometimes used to define the dual map.
The Kraus decomposition of the adjoint of a Hermitian-preserving map $\Phi$ is then $\Phi^\dagger = \sum_i \alpha_i K_i^\dagger (\cdot) K_i$. 
It is well known that $\Phi$ is Hermitian-preserving/positive/completely positive if and only if the same is true for $\Phi^\dagger$, and $\Phi$ is trace-preserving if and only if $\Phi^\dagger$ is unital, that is, $\Phi^\dagger({\bf1})={\bf1}$.

\subsection{Entanglement witnesses \& positive maps}\label{sec_EW_pos}

Next, let us recall well-known concepts about bipartite entanglement and its detection \cite{HHHH07,GT09,CS14}.
Given some $\rho\in\mathbb C^{m\times m}\otimes\mathbb C^{n\times n}$, $\rho\geq 0$ we say $\rho$ is separable if there exist $\ell\in\mathbb N$ and $\rho_1,\ldots,\rho_\ell\in\mathbb C^{m\times m}$, $\sigma_1,\ldots,\sigma_\ell\in\mathbb C^{n\times n}$ all positive semi-definite such that $\rho=\sum_{j=1}^\ell\rho_j\otimes\sigma_j$. Equivalently, $\rho=\sum_j|x_j\rangle\langle x_j|\otimes|y_j\rangle\langle y_j|$ for some $\{x_j\}_j\subset\mathbb C^m$, $\{y_j\}_j\subset\mathbb C^n$ \cite[Coro.~6.7]{Watrous18}. If $\rho$ is not separable we say $\rho$ is entangled.

This gives rise to so-called entanglement-breaking maps, that is, to completely positive maps $\Psi$ for which $({\rm id}_R\otimes\Psi)(\rho)$ (equivalently: $(\Psi\otimes{\rm id}_R)(\rho)$) is separable for all reference systems $R$ and all \mbox{$\rho\geq 0$} \cite{HSR03}. 
Indeed---similar to complete positivity---$\Psi$ is entanglement-breaking if and only if $\mathsf C(\Psi)$ is separable \cite[Prop.~6.22]{Holevo12ed2}.
If such $\Psi$ is additionally trace-preserving, then it is called ``entanglement-breaking channel''.
A prominent example here is the depolarizing 
channel \mbox{$\Psi_D:\mathbb C^{n\times n}\to\mathbb C^{n\times n}$}, \mbox{$X\mapsto (1-p)X+p\,{\rm tr}(X)\frac{\bf1}{n}$}, \mbox{$p\in[0,1]$} which is known to be entanglement-breaking for all $p\in[ \frac{n}{n+1} , 1 ]$ \cite[Prop.~6.40]{Holevo12ed2}.
In particular, for $p=\frac{n}{n+1}$ one has that
\begin{equation}\label{eq:EB_example}
\mathbb C^{n\times n}\ni X\mapsto \frac{1}{n+1}\big(X+{\rm tr}(X){\bf1}\big)
\end{equation}
is an entanglement-breaking channel (cf.~also \cite[Thm.~7.5.4]{Stormer13}).

Next, a convenient tool for detecting entanglement are entanglement witnesses. These are block-positive operators $W \in \mathbb C^{m\times m}\otimes\mathbb C^{n\times n}$ for which there exists $\sigma\geq 0$ such that ${\rm tr}(W\sigma)<0$; recall that $W$ is called block-positive if it satisfies 
     $
     \langle x\otimes y | W |x\otimes y \rangle \geq 0 
     $
for all
$x,y$. Another definition of block-positivity---which is readily seen to be equivalent via the spectral decomposition---is that ${\rm tr}(W(P\otimes Q))\geq 0$ for all $P,Q\geq 0$. This way one sees that all block-positive operators---and thus all entanglement witnesses $W$---satisfy ${\rm tr}(W\rho)\geq 0$ for all $\rho$ separable.
In particular, this shows that ${\rm tr}_1(W),{\rm tr}_2(W)\geq 0$ for all $W$ block-positive, because ${\rm tr}(P{\rm tr}_1(W))={\rm tr}(({\bf1}\otimes P)W)\geq 0$ for all $P\geq 0$. 
This is also why all $W$ block-positive satisfy ${\rm tr}(W)\geq 0$, with equality if and only if $W=0$.

Entanglement witnesses are well-known to have an equivalent formulation using (completely) positive maps via
CJ \cite[Sec. IV.B.3]{HHHH07}, \cite{GT09}: $W \in \mathbb C^{m\times m}\otimes\mathbb C^{n\times n}$ is block-positive (an entanglement witness) if and only if $W=\mathsf C(\Phi^\dagger)$ 
for some positive (but not completely positive) map $\Phi:\mathbb C^{n\times n}\to\mathbb C^{m\times m}$.
This connection shows that, equivalently, $W$ is block-positive if and only if $({\rm id}\otimes\Psi)(W)\geq 0$ for all entanglement-breaking channels $\Psi:\mathbb C^{n\times n}\to\mathbb C^{m\times m}$  \cite[Eq.~(16.62)]{Bengtsson17}.
Either way this yields the following famous characterization of separability \cite{HHH96,Stormer86}: 
\begin{lemma}
    Given any $\rho\in\mathbb C^{m\times m}\otimes\mathbb C^{n\times n}$ positive semi-definite, the following statements are equivalent.
    \begin{itemize}
        \item[(i)] $\rho$ is separable.
        \item[(ii)] $({\rm id}\otimes\Phi)(\rho)\geq 0$ for all $\Phi:\mathbb C^{n\times n}\to\mathbb C^{k\times k}$ positive (equivalently: $(\Phi^\dagger\otimes{\rm id})(\rho)\geq 0$ for all \mbox{$\Phi:\mathbb C^{k\times k}\to\mathbb C^{m\times m}$} positive) for all $k\in\mathbb N$.
        \item[(iii)] $({\rm id}\otimes\Phi)(\rho)\geq 0$ (equivalently: $(\Phi^\dagger\otimes{\rm id})(\rho)\geq 0$) for all positive maps $\Phi:\mathbb C^{n\times n}\to\mathbb C^{m\times m}$.
        \item[(iv)] ${\rm tr}(W\rho)\geq 0$ for all block-positive operators (resp., all entanglement witnesses) $W\in\mathbb C^{m\times m}\otimes\mathbb C^{n\times n}$.
    \end{itemize}
\end{lemma}
\noindent At first glance (ii) may look like a more bloated and hence redundant version of (iii).
However, the framework of (iii) does not allow for such fundamental operations as the partial transpose (if $m\neq n$), which is what the varying dimension $k$ in (ii) takes care of.

While the set of all witnesses has the same entanglement detection power as the set of all positive maps, a single positive map $\Phi$ is always stronger than the induced witness $W=\mathsf C(\Phi^\dagger)$ in the sense that
\begin{equation}
\label{eq:witness_to_channel_inequality}
    {\rm tr}(W \rho) <0 \ \Rightarrow \ ({\rm id}\otimes\Phi)(\rho)\not\geq 0
\end{equation}
for all $\rho$, while the reverse implication does not hold for all states \cite{GT09}. Note that the ${}^\dagger$ in the witness has to be included so $W\rho$ is well defined, because else the dimensions do not match if $m\neq n$.
Now the reason why~\eqref{eq:witness_to_channel_inequality} holds is simple:
\begin{align*}
    0 > {\rm tr}(\mathsf C(\Phi^\dagger) \rho)  & 
    = {\rm tr}(({\rm id} \otimes\Phi^\dagger)(|\Gamma\rangle\langle\Gamma|) \rho) \\ & 
    = {\rm tr}(|\Gamma\rangle\langle\Gamma| ({\rm id} \otimes \Phi)(\rho)) \\ & 
    = \langle\Gamma| ({\rm id} \otimes \Phi)(\rho) |\Gamma\rangle \,,
\end{align*}
and $({\rm id} \otimes \Phi)(\rho)$ admitting a negative expectation value implies that it cannot be positive.
Equivalently, 
$
{\rm tr}(W^T \rho) <0 $ with $W=\mathsf C(\Phi^\dagger)$ implies $ (\Phi^\dagger\otimes{\rm id})(\rho)\not\geq 0
$
due to the readily verified identity
\begin{equation}\label{eq:Choi_adjoint}
    \big(\mathsf C(\cdot)\big)^T = \mathbb F\,\mathsf C\big((\cdot)^\dagger\big)\mathbb F\,.
\end{equation}
Despite all that, we want to stress that any individual positive-map criterion $({\rm id}\otimes\Phi)(\rho)\geq 0$---which is necessary for $\rho$ to be separable---is known to be equivalent to a continuous family of witnesses induced by the maps $\{X^\dagger\Phi(\cdot)X:X\}$ \cite{HHHH07,HE02}: More precisely, given some $\Phi:\mathbb C^{n\times n}\to\mathbb C^{m\times m}$ positive and $\rho\in\mathbb C^{m\times m}\otimes\mathbb C^{n\times n}$, $\rho\geq 0$ one has $({\rm id}\otimes\Phi)(\rho)\geq 0$ if and only if
\begin{align}
{\rm tr}\big( (X\otimes{\bf1})\mathsf C(\Phi^\dagger)(X^\dagger\otimes{\bf1}) \rho \big)\geq 0\label{eq:P_witness_family}
\end{align}
for all $X\in\mathbb C^{m\times m}$. The key here is the simple
identity
${\rm tr}( (X\otimes{\bf1})\mathsf C(\Phi^\dagger)(X^\dagger\otimes{\bf1}) \rho)=\langle{\rm vec}(X^T)|({\rm id}\otimes\Phi)(\rho)|{\rm vec}(X^T)\rangle$; here ${\rm vec}$ is the usual (column-) vectorization ${\rm vec}=({\bf 1}\otimes (\cdot))|\Gamma\rangle=((\cdot)^T\otimes{\bf1})|\Gamma\rangle$ which, in particular, satisfies the key identity ${\rm vec}(ABC)=(C^T\otimes A){\rm vec}(B)$ for all compatible matrices $A,B,C$ \cite[Ch.~4.2 ff.]{HJ2}.
Also note that not all $X$ are needed in~\eqref{eq:P_witness_family} but only those of full rank, i.e., $({\rm id}\otimes\Phi)(\rho)\geq 0$ if and only if~\eqref{eq:P_witness_family} holds for all $X$ of full rank (as can be seen via a simple continuity argument).

The final concept we recap is optimality of entanglement witnesses \cite{HHHH07,GT09}, cf.~also Fig.~\ref{fig:witness_optimality} below (beware footnote~%
\cite{footnote_fig1}%
):
Given a witness $W$ define the set \mbox{$D_W:=\{\rho:{\rm tr}(W\rho)<0\}$} of all states $\rho$ which $W$ detects. Then a witness $W'$ is said to be better than $W$ if it detects more entangled states, in the sense that $D_W\subsetneq D_{W'}$. Now a witness $W$ is called optimal if there exists no witness $W'$ better than $W$.
It is known~\cite[Thm.~1]{LKCH00},~\cite[Sec.~2.5.2]{GT09} that a witness $W$ is optimal if and only if there exists no $P\geq 0$, $P\neq 0$ such that $W-P$ is block-positive.
This implies the important fact that if any witness derived from a positive map is optimal, then so is the entire equivalent family of witnesses (cf.~Eq.~\eqref{eq:P_witness_family} ff.):
\begin{lemma}\label{lemma_optimal_family}
    Given any entanglement witness \mbox{$W\in\mathbb C^{m\times m}\otimes\mathbb C^{n\times n}$} the following statements are equivalent:
    \begin{itemize}
        \item[(i)] $W$ is optimal.
        \item[(ii)] There exists $Y\in\mathbb C^{m\times m}$ of full rank such that \mbox{$(Y\otimes{\bf1})W(Y^\dagger\otimes{\bf1})$} is optimal.
        \item[(iii)] For all $X\in\mathbb C^{m\times m}$ of full rank, $(X\otimes{\bf1})W(X^\dagger\otimes{\bf1})$ is optimal.
    \end{itemize}
\end{lemma}
\begin{proof}
``(iii) $\Rightarrow$ (i) $\Rightarrow$ (ii)'': Trivial.
``(ii) $\Rightarrow$ (iii)'': We argue by contrapositive. Assume there exists $X$ of full rank such that $(X\otimes{\bf1})W(X^\dagger\otimes{\bf1})$ is not optimal. Hence there exists $P\geq 0$, $P\neq 0$ such that $(X\otimes{\bf1})W(X^\dagger\otimes{\bf1})-P$ is block-positive.
Thus for all $Y$ of full rank, 
$$
(Y\otimes{\bf1})W(Y^\dagger\otimes{\bf1})-(YX^{-1}\otimes{\bf1})P((YX^{-1})^\dagger\otimes{\bf1})
$$
is block-positive, too, because block-positivity is preserved under such local transformations. Moreover, $(YX^{-1}\otimes{\bf1})P((YX^{-1})^\dagger\otimes{\bf1})$ is non-zero (because $X,Y$ are invertible) and positive which lets us conclude that $(Y\otimes{\bf1})W(Y^\dagger\otimes{\bf1})$ cannot be optimal for any $Y$ of full rank. This concludes the proof.
\end{proof}

\begin{figure}
    \centering
    \includegraphics[scale=1]{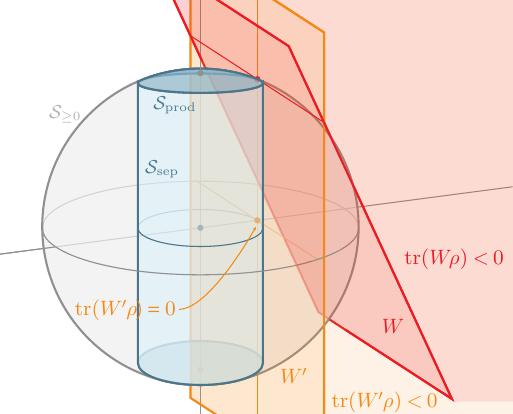}
    \caption{
        Graphical representation of the effect of entanglement witnesses and their optimality (recall from footnote~\cite{footnote_fig1} that the geometry in this figure is only illustrative). 
        The grey sphere ($\mathcal{S}_{\geq 0}$) represents the set of quantum states,
        while $\mathcal{S}_{\rm prod}$ (the north and south caps of the sphere) stands for the set of pure product states. 
        Then $\mathcal{S}_{\rm sep}$ (the cylinder-like shaped convex hull of the caps) is the set of separable states.
        Each witness $W$ defines a hyperplane ${\rm tr}(W\rho)=0$ separating detected entangled states (${\rm tr}(W\rho)<0$) from the rest (${\rm tr}(W\rho)\geq 0$).
        The witness $W$ (red) ``touches'' the set of separable states (the state marked as a red dot is separable but fulfills ${\rm tr}(W\rho)=0$) so $W$ is weakly optimal. 
        It is clearly sub-optimal, as all states it detects (states within the sphere in the area shaded in red, right from the red plane) are also detected by $W'$. 
        The witness $W'$ on the other hand satisfies ${\rm tr}(W'\rho)=0$ for a whole range of states (the orange line on the cylinder which is also in the orange hyperplane); in particular, this is true for a full-rank separable state (orange dot) which is equivalent to the spanning property (Thm.~\ref{thm_main_new}) and hence sufficient for optimality of $W'$.
        }
    \label{fig:witness_optimality}
\end{figure}

A sufficient criterion for optimality of a witness \mbox{$W\in\mathbb C^{m\times m}\otimes\mathbb C^{n\times n}$} (which is not necessary in general \cite{ATL11,ASL11} except for some special cases \cite{ATL11}) is $\dim ({\rm span} \{|a_i\rangle\otimes|b_i\rangle : \langle a_i\otimes b_i|W|a_i\otimes b_i\rangle=0\}) = mn$ which is called ``spanning property'' \cite{LKCH00}.
For characterizations of the spanning property we refer to Thm.~\ref{thm_main_new} below.
On the other hand, a necessary (but not sufficient) criterion for optimality of $W$ is that the hyperplane spanned by $W$ touches the set of separable states, i.e., ${\rm tr}(W\rho)=0$ for some separable state $\rho$. This is also known as ``weak optimality''; equivalently, there exist pure states $x,y$ such that $\langle x\otimes y|W|x\otimes y\rangle=0$ \cite{GT09,WXCS15}.
For a visual representation of the relation between different known (sufficient and/or necessary) conditions on entanglement witnesses, as well as how our results relate to these and complement the picture, see Fig.~\ref{fig:overview_diagram}.

For further reading, there is a characterization of optimality in terms of the subspace orthogonal to the one from the spanning property \cite{LKCH00} or in terms of generalized Kraus operators \cite{QH12}. There also exist specialized optimality criteria for the decomposable case \cite{LKCH00,Kye12}, and another relevant notion here is the so-called ``exposedness'' property for positive maps \cite{HK11}.

\begin{figure*}
    \centering
    \includegraphics[scale=1]{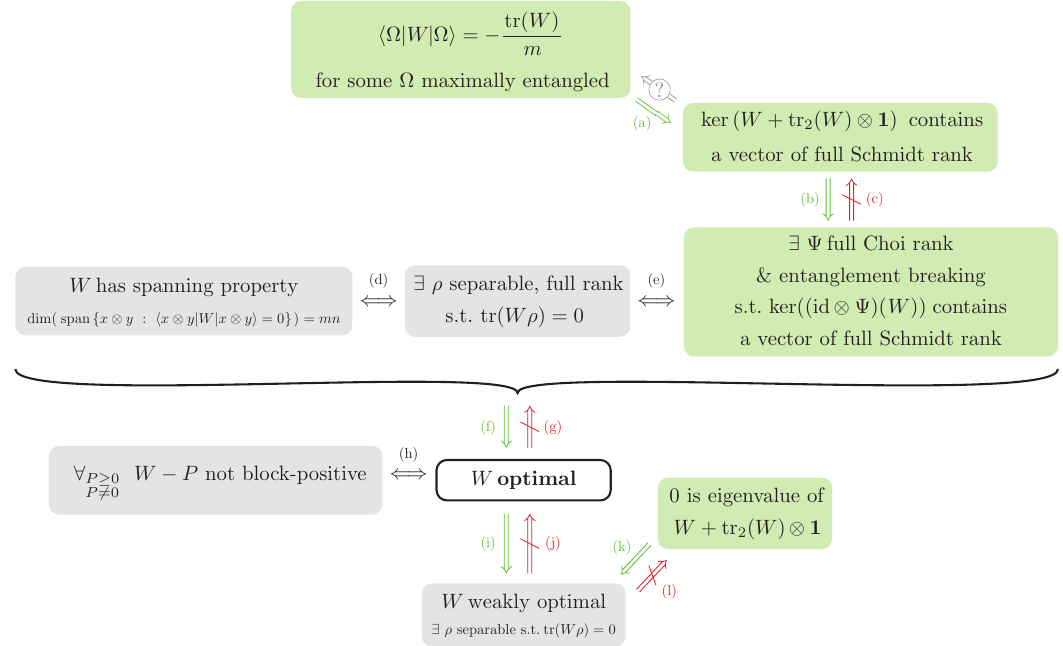}
    \caption{
    Overview of
    optimality conditions on entanglement witnesses $W\in\mathbb C^{m\times m}\otimes\mathbb C^{n\times n}$ where, without loss of generality,
    $m\leq n$. 
    Green boxes are new results proven in this work.
    (a) Proof of Coro.~\ref{coro:trace_bound_from_max_entangled},
    (b) Proof of Thm.~\ref{thm:full_Schmidt_rank_to_optimal},
    (c) Rem.~\ref{rem_kernel_weaker_than_spanning}~(i),
    (d) \& (e) Thm.~\ref{thm_main_new},
    (f) \& (h) Lewenstein et al.~\cite{LKCH00},
    (g) Augusiak et al.~\cite{ATL11,ASL11}, see also Choi et al.~\cite{CK12},
    (i) Gühne et al.~\cite{GT09},
    (j) $W=\mathbb F+|00\rangle\langle 00|$, because $W-|00\rangle\langle 00|$ is a witness (hence $W$ is not optimal), but $\langle 01|W|01\rangle=0$ (weakly optimal),
    (k) Coro.~\ref{cor:zero_eigval_to_weak_opt},
    (l) $W=\mathbb F+|00\rangle\langle 00|$ is weakly optimal but $0$ is not an eigenvalue of $W+{\rm tr}_2(W)\otimes{\bf1}$
    }
    \label{fig:overview_diagram}
\end{figure*}

Before moving on to our main results let us collect some common witnesses and relate them to the corresponding positive but not completely positive maps.
Later we will use
these examples to illustrate our optimality criteria. 
First, the transpose map is one of the most famous examples and induces the PPT (positive partial transpose) criterion $({\rm id}\otimes{}^T)(\rho)\geq 0$~\cite{peresSeparabilityCriterionDensity1996}.
Although in general it is only a necessary (but not sufficient) condition for separability of $\rho$, for system sizes $\mathbb C^2 \otimes \mathbb C^2$ and $\mathbb C^2 \otimes \mathbb C^3$ this is equivalent to $\rho$ being separable~\cite{HHH96}.
The transpose is self-adjoint, so
the witness obtained by direct application of Choi is the $\mathsf{SWAP}$ operator \mbox{$W_{{}^T} = \mathsf C ({}^T) = \mathbb F$}, and its optimality is easily established via the spanning property. 
The PPT criterion obtained from the transpose map is then equivalent to the family of witnesses $(X\otimes{\bf1})\mathbb F(X^\dagger\otimes{\bf1})=(X\otimes X^\dagger)\mathbb F$.

The map $\Phi:\mathbb C^{n\times n}\to\mathbb C^{n\times n}$, $\Phi(X):={\rm tr}(X){\bf1}-X$ is known as reduction map~\cite{HH99}. 
It, too, is self-adjoint so the induced witness is 
$$W_{\mathrm{red}} = \mathsf C (\Phi) = {\bf 1}_{n} \otimes {\bf 1}_{n} - | \Gamma \rangle \langle \Gamma |\,.$$
Its optimality was shown, e.g., in \cite{CP11}.
Next is the Choi map $\Phi:\mathbb C^{3\times 3}\to\mathbb C^{3\times 3}$---which was the first example of an indecomposable map (i.e., cannot be written as $\Phi_1+\Phi_2\circ{}^T$ for any $\Phi_1,\Phi_2$ completely positive)---defined via
$$
\Phi(X) \coloneqq 2 {\rm tr}(X) {\bf1}_3 - 2 {\rm diag} (X_{33},X_{11},X_{22}) - X
$$
\cite{Choi75b} (cf.~also \cite[p.~301]{Bengtsson17} for a list of generalizations).
Its adjoint map reads
$
    \Phi^\dagger(X) = 2\cdot {\rm tr}(X) {\bf1}_3 - 2 {\rm diag} (X_{22},X_{33},X_{11}) - X 
$
so the induced witness is
\begin{equation}\label{eq:Choi_witness}
    W_{\rm Choi} 
    = 2\cdot {\bf 1} \otimes {\bf 1} - 2(| 02 \rangle \langle 02 | + | 10 \rangle \langle 10 | + | 21 \rangle \langle 21 |)- | \Gamma \rangle \langle \Gamma |
\end{equation}
Beware that---unlike the previous witnesses---$W_{\rm Choi}$ is not optimal as $W_{\rm Choi}-(|01\rangle\langle 01|+|12\rangle\langle 12|+|20\rangle\langle 20|)$ is still block-positive (in the language of \cite{CKL92} this better witness corresponds to the positive map $\Phi[2,1,0]$ which is now, in fact, optimal \cite[Sec.~7.2]{CS14} without having the spanning property \cite{CK12}).

Another well known criterion comes from the Breuer-Hall map \cite{Breuer06b,Breuer06,Hall06}
which in the most general form \cite[Eq.~(15)]{AS09} is defined as
\begin{equation}\label{eq:gen_BH}
\begin{split}
    \Phi:\mathbb C^{2n\times 2n} & \to \mathbb C^{2n\times 2n} \\
    X & \mapsto{\rm tr}(X){\bf1}-X-UX^TU^\dagger
\end{split}    
\end{equation}
with $U\in\mathbb C^{2n\times 2n}$ sub-unitary ($U^\dagger U\leq{\bf1}$) and antisymmetric ($U^T=-U$), e.g., $U={\bf1}_n\otimes\sigma_y$.
The corresponding entanglement witness
\begin{equation}\label{eq:BreuerHallWitness}
    W_{\rm BH} = \mathsf C(\Phi^\dagger) = {\bf 1} \otimes {\bf 1} - | \Gamma \rangle \langle \Gamma | - ({\bf1}\otimes U) \mathbb F ({\bf1}\otimes U^\dagger)
\end{equation}
is known to be non-decomposable and optimal for $U$ unitary \cite{Breuer06b}, 
and using our new criterion we will easily show that it is optimal even for all $U$ sub-unitary. Note that the reduction map is just a special case of this generalized Breuer-Hall map if $U=0$.

Finally, choosing $n=2$ and $U={\bf1}_2\otimes\sigma_y$ in Eq.~\eqref{eq:gen_BH} reproduces the Robertson map \cite{Robertson83,Robertson85,CS12}, two generalizations of which read as follows \cite{CPS09,CP10}:
$\Phi_{2n}:\mathbb C^{2n\times 2n}\to \mathbb C^{2n\times 2n}$ which maps $X$ to $\frac1n$ times
\begin{equation}\label{eq:gen_Rob_1}
\begin{pmatrix}
    {\rm tr}(X_{22}){\bf1}&-(X_{12}+{\rm tr}(X_{21}){\bf1}-X_{21})\\
    -(X_{21}+{\rm tr}(X_{12}){\bf1}-X_{12})&{\rm tr}(X_{11}){\bf1}
\end{pmatrix}
\end{equation}
as well as
\begin{equation}\label{eq:gen_Rob_2}
\begin{split}
    \Phi_{4n}:\mathbb C^{4n\times 4n}&\to \mathbb C^{4n\times 4n}\\
    \begin{pmatrix}
        X_{11}&X_{12}\\X_{21}&X_{22}
    \end{pmatrix}&\mapsto \begin{pmatrix}
    {\rm tr}(X_{22}){\bf1}&-X_{12}-UX_{21}^TU^\dagger\\
    -X_{21}-UX_{12}^TU^\dagger&{\rm tr}(X_{11}){\bf1}
\end{pmatrix}
\end{split}
\end{equation}
where $U\in\mathbb C^{2n\times 2n}$ any antisymmetric unitary.
Notably, both are non-decomposable positive maps and their associated witnesses are optimal.
While there are many more examples of positive maps
in the literature (Woronowicz map \cite{Woronowicz76}, atomic positive maps \cite{Ha03}, positive maps from mutually unbiased bases \cite{CSW18} and measurements \cite{LLLFW20,SC21}, etc.)---cf.~also \cite[Sec.~4]{CK07}, \cite{GF13} and references therein---we will leave it at that and instead come to our main results.

\section{Main Results}\label{sec_mainres}

\subsection{New sufficient optimality criteria}\label{sec_opt}
Let us start with a novel characterization of the spanning property; we note that we believe the equivalence of (i) and (ii) in the following theorem to be known, but we could not locate a reference which spells this out explicitly so we will provide a proof for that as well, for the sake of completeness.

\begin{theorem}\label{thm_main_new}
    Given any $W\in\mathbb C^{m\times m}\otimes\mathbb C^{n\times n}$ block-positive, the following statements are equivalent.
    \begin{itemize}
        \item[(i)] $W$ has the spanning property.
        \item[(ii)] There exists a separable state $\rho\in\mathbb C^{m\times m}\otimes\mathbb C^{n\times n}$ of full rank such that ${\rm tr}(W\rho)=0$.
        \item[(iii)] There exists $\Psi:\mathbb C^{n\times n}\to\mathbb C^{m\times m}$ entanglement-breaking such that $\Psi$ has full Choi rank and ${\rm tr}(W\mathsf C(\Psi^\dagger))=0$.
        \item[(iv)] One of the following statements holds:
        \begin{itemize}
            \item[(a)] There exist $k\geq m$ and $\Psi:\mathbb C^{n\times n}\to\mathbb C^{k\times k}$ entanglement-breaking such that $\Psi$ has full Choi rank and $\ker(({\rm id}\otimes\Psi)(W))$ contains a vector of Schmidt rank $m$.
            \item[(b)] There exist $k\geq n$ and $\Psi:\mathbb C^{k\times k}\to\mathbb C^{m\times m}$ entanglement-breaking such that $\Psi$ has full Choi rank and $\ker((\Psi^\dagger\otimes{\rm id})(W))$ contains a vector of Schmidt rank $n$.
        \end{itemize}
    \end{itemize}
\end{theorem}
\begin{proof}
    ``(ii) $\Rightarrow$ (i)'': As mentioned previously, because $\rho$ is separable it can be written as $\rho=\sum_j\lambda_j|x_j\rangle\langle x_j|\otimes|y_j\rangle\langle y_j|$ for some $\lambda_j>0$, $\sum_j\lambda_j=1$ and some unit vectors $\{x_j\}_j,\{y_j\}_j$. Now by assumption $0={\rm tr}(W\rho)=\sum_j\lambda_j\langle x_j\otimes y_j|W|x_j\otimes y_j\rangle$; but $W$ is block-positive so all summands are non-negative, meaning---as $\lambda_j>0$---this sum vanishes if and only if $\langle x_j\otimes y_j|W|x_j\otimes y_j\rangle=0$ for all $j$. Thus all that is left to show is that $\{x_j\otimes y_j:j\}$ spans the full space: because $\rho$ is full rank
    \begin{align*}
        \mathbb C^m\otimes\mathbb C^n=\operatorname{range}\rho&=\{\rho|x\rangle:x\in\mathbb C^m\otimes\mathbb C^n\}\\
        &= \Big\{\sum_j\lambda_j\langle x_j\otimes y_j|x\rangle |x_j\otimes y_j\rangle:x\Big\}\\
        &\subseteq \Big\{\sum_jc_j|x_j\otimes y_j\rangle:c_j\in\mathbb C\Big\}\\
        &={\rm span}\{|x_j\otimes y_j\rangle:j\}\subseteq\mathbb C^m\otimes\mathbb C^n\,.
    \end{align*}
    
     ``(i) $\Rightarrow$ (ii)'': Assume $W$ has the spanning property, i.e., there exist $mn$ linearly independent pure product states $a_i\otimes b_i$ such that $\langle a_i\otimes b_i|W|a_i\otimes b_i\rangle=0$ for all $i$. In particular this implies
    $$
    \{a_i\otimes b_i:i\}^\perp=({\rm span}\{a_i\otimes b_i:i\})^\perp=(\mathbb C^m\otimes\mathbb C^n)^\perp=\{0\}
    $$
    which means that for all $x\neq 0$ there exists $i$ such that $\langle x|a_i\otimes b_i\rangle\neq 0$.
    Now to construct the separable full-rank state we simply define $\rho=\frac1C\sum_i|a_i\rangle\langle a_i|\otimes|b_i\rangle\langle b_i|$ with $C=\sum_i\|a_i\|^2\|b_i\|^2$ a normalization constant. Certainly, $\rho$ is separable and ${\rm tr}(W\rho)=\frac1C\sum_i\langle a_i\otimes b_i|W|a_i\otimes b_i\rangle=0$. To see that $\rho$ is full rank let any $x\neq 0$ be given and observe that $\langle x|\rho|x\rangle=\frac1C\sum_i|\langle a_i\otimes b_i|x\rangle|^2>0$ by our previous argument.

    ``(ii) $\Leftrightarrow$ (iii)'': As explained in Sec.~\ref{sec_prelim}, the CJ
    isomorphism serves as a one-to-one translation between positive semi-definite, separable \& full rank (on the level of bipartite operators), and completely positive, entanglement-breaking \& full Choi rank (on the level of linear maps), respectively. Hence given $\rho$ choose $\Psi=\mathsf C^{-1}(\rho)^\dagger$, and given $\Psi$ choose $\rho=\frac{\mathsf C(\Psi^\dagger)}{{\rm tr}(\mathsf C(\Psi^\dagger))}$.
    Here we use that $\Psi$ is entanglement-breaking ($\mathsf C(\Psi)$ is separable) if and only if $\Psi^\dagger$ is entanglement-breaking which follows at once from Eq.~\eqref{eq:Choi_adjoint}.

    ``(iii) $\Rightarrow$ (iv)'': First assume $n\geq m$; we want to show that (iii) implies (iv),(a) for $k=m$. We compute
    \begin{align*}
        0={\rm tr}(W\mathsf C(\Psi^\dagger))&={\rm tr}(W({\rm id}\otimes\Psi^\dagger)(|\Gamma\rangle\langle\Gamma|))\\
        &={\rm tr}(({\rm id}\otimes\Psi)(W)|\Gamma\rangle\langle\Gamma|)\\
        &=\langle\Gamma|({\rm id}\otimes\Psi)(W)|\Gamma\rangle\,.
    \end{align*}
    As stated previously,
    $({\rm id}\otimes\Psi)(W)\in\mathbb C^{m\times m}\otimes\mathbb C^{m\times m}$ is positive semi-definite
    because $\Psi$ is entanglement-breaking and $W$ is block-positive.
    Hence
    $$
    0=\langle\Gamma|({\rm id}\otimes\Psi)(W)|\Gamma\rangle=\Big\|\sqrt{({\rm id}\otimes\Psi)(W)}|\Gamma\rangle\Big\|^2
    $$
    implies that $|\Gamma\rangle\in\mathbb C^m\otimes\mathbb C^m$ (i.e., Schmidt rank $m$)
    is an element of $\ker(({\rm id}\otimes\Psi)(W))$, as desired.
    Similarly, if $m\geq n$ then (iv),(b) holds for $k=n$ and $\tilde\Psi$ such that $\mathsf C(\tilde\Psi^\dagger)=\mathsf C(\Psi^\dagger)^T$ (as argued before $\tilde\Psi$ is still entanglement-breaking and has full Choi rank): using Eq.~\eqref{eq:Choi_adjoint} $$0={\rm tr}(W\mathsf C(\tilde\Psi^\dagger)^T)={\rm tr}(W((\tilde\Psi\otimes{\rm id})(|\Gamma\rangle\langle\Gamma)))$$
    so as before $\mathbb C^n\otimes\mathbb C^n\ni|\Gamma\rangle\in\ker((\tilde\Psi^\dagger\otimes{\rm id})(W))$ (Schmidt rank $n$).

    ``(iv) $\Rightarrow$ (iii)'': First assume that (iv),(a) holds, i.e., there exists $k\geq m$, some $\Psi:\mathbb C^{n\times n}\to\mathbb C^{k\times k}$ entanglement-breaking of full Choi rank, and---using the previously introduced concept of vectorization---some $X\in\mathbb C^{k\times m}$, ${\rm rank}(X)=m$ (by definition of the Schmidt rank) such that ${\rm vec}(X)\in \ker(({\rm id}\otimes\Psi)(W))\subseteq\mathbb C^m\otimes\mathbb C^k$.
    We claim that $\tilde\Psi:=X^\dagger\Psi(\cdot)X:\mathbb C^{n\times n}\to\mathbb C^{m\times m}$ satisfies (iii). First,
    \begin{align*}
        {\rm tr}(W\mathsf C(\tilde\Psi^\dagger))&={\rm tr}\big(W\big(({\rm id}\otimes\Psi^\dagger)(\mathsf C(X(\cdot)X^\dagger))\big)\big)\\
        &={\rm tr}\big(W\big(({\rm id}\otimes\Psi^\dagger)(|{\rm vec}(X)\rangle\langle {\rm vec}(X)|)\big)\big)\\
        &=\langle {\rm vec}(X)|({\rm id}\otimes\Psi)(W)|{\rm vec}(X)\rangle=0\,.
    \end{align*}
    Next, $\mathsf C(\tilde\Psi)=({\bf1}\otimes X)^\dagger\mathsf C(\Psi)({\bf1}\otimes X)$ immediately shows that $\tilde\Psi$ is entanglement-breaking (because $\Psi$ is). Thus all that remains is to show that $\ker(\mathsf C(\tilde\Psi))=\{0\}$.
    For this let $Y\in\mathbb C^{m\times n}$ be given such that ${\rm vec}(Y)\in \ker(\mathsf C(\tilde\Psi))$. A straightforward computation shows
    \begin{align*}
        0&=\langle{\rm vec}(Y)|\mathsf C(\tilde\Psi)|{\rm vec}(Y)\rangle\\
        &=\langle({\bf1}\otimes Y)\Gamma|({\bf1}\otimes X^\dagger)\mathsf C(\Psi)({\bf1}\otimes X)|({\bf1}\otimes Y)\Gamma\rangle\\
        &=\langle{\rm vec}(XY)|\mathsf C(\Psi)|{\rm vec}(XY)\rangle\,.
    \end{align*}
    We know $\mathsf C(\Psi)>0$ by assumption so ${\rm vec}(XY)$ can only be zero---but ${\rm vec}$ is an isomorphism so this even shows $XY=0$.
    The final step is to deduce $Y=0$ from this. Indeed, given any $z\in\mathbb C^n$ we know $0=X(Yz)$, i.e., $Yz\in\ker(X)$. But the rank-nullity theorem \cite[Eq.~(1.54)]{Watrous18} implies $m=\dim\ker(X)+{\rm rank}(X)=\dim\ker(X)+m$ so $\ker(X)=\{0\}$. Hence $Yz=0$; but $z$ was arbitrary so $Y=0$ which---because ${\rm vec}(Y)$ was an arbitrary kernel element of $\mathsf C(\tilde\Psi)$---shows $\mathsf C(\tilde\Psi)>0$ as claimed.

    For (iv),(b) $\Rightarrow$ (iii) one argues analogously, where now $X\in\mathbb C^{n\times k}$ is the (un-vectorized) full-rank kernel element of $(\Psi^\dagger\otimes{\rm id})(W)$, and $\tilde\Psi:=( \Psi(X^T(\cdot)^T(X^T)^\dagger )^T:\mathbb C^{n\times n}\to\mathbb C^{m\times m}$ is the full Choi rank entanglement-breaking map we are looking for. Again using Eq.~\eqref{eq:Choi_adjoint}, the rest follows from 
    $\mathsf C(\tilde\Psi^\dagger)=(\Psi\otimes{\rm id})(|{\rm vec}(X)\rangle\langle {\rm vec}(X)|)$
    as well as
    $\mathsf C(\tilde\Psi)=(X^T\otimes{\bf1})^\dagger\mathsf C(\Psi)^T(X^T\otimes{\bf1})$.
\end{proof}

\noindent For a simple example which nicely illustrates how the construction in the proof of Thm.~\ref{thm_main_new} (i) $\Rightarrow$ (ii) works we refer to Appendix~\ref{appA}. There one also sees that the spanning property yields not just a single state $\rho$ on which $W$ vanishes, but an entire face of the set of separable states.

While weak optimality of witnesses is not the main focus of this work, let us still assess how our first theorem has to be adjusted if one wants to characterize weak optimality using entanglement-breaking maps:

\begin{remark}\label{rem_weak_opt}
The equivalence from Thm.~\ref{thm_main_new} becomes less obvious when replacing the spanning property in (i) by weak optimality, in the following sense: There exist block-positive matrices $W\in\mathbb C^{m\times m}\otimes\mathbb C^{n\times n}$ which are not weakly optimal,
but there exists a non-zero entanglement-breaking $\Psi:\mathbb C^{m\times m}\to\mathbb C^{n\times n}$ such that $0$ is an eigenvalue of $({\rm id}\otimes\Psi^\dagger)(W)$ and of $(\Psi\otimes{\rm id})(W)$ (to a common eigenvector, even).
A simple example here is $W={\bf1}\in\mathbb C^{4\times 4}$ and $\Psi={\rm tr}(\cdot)|0\rangle\langle0|$: $W$ is clearly not weakly optimal, but $({\rm id}\otimes\Psi)(W)=2\cdot{\bf1}\otimes|0\rangle\langle 0|$ and $(\Psi^\dagger\otimes{\rm id})=|0\rangle\langle 0|\otimes{\bf1}$
both have $|11\rangle$ in their kernel.
In fact, it is easy to see that weak optimality is equivalent to 
the existence of $\Psi$ entanglement-breaking and non-zero $x\in\ker(({\rm id}\otimes\Psi^\dagger)(W))\cup \ker((\Psi\otimes{\rm id})(W))$ such that, additionally, $({\rm id}\otimes\Psi)(|x\rangle\langle x|)$ or $(\Psi^\dagger\otimes{\rm id})(|x\rangle\langle x|)$ does not vanish.

The reason we did not encounter this problem in Thm.~\ref{thm_main_new} is that either the full Choi rank condition on $\Psi$ or the full Schmidt rank condition on $x$ are sufficient to guarantee $0\neq ({\rm id}\otimes\Psi)(|x\rangle\langle x|)=(X^T\otimes{\bf1})\mathsf C(\Phi)(X^T\otimes{\bf1})^\dagger$ (and similarly for $(\Psi^\dagger\otimes{\rm id})(|x\rangle\langle x|)$). In particular, this shows that given some $\Psi$ entanglement-breaking with full Choi rank, if $0$ is an eigenvalue of $({\rm id}\otimes\Psi^\dagger)(W)$ or $(\Psi\otimes{\rm id})(W)$, then $W$ is weakly optimal.
\end{remark}

Now at first glance it looks like Thm.~\ref{thm_main_new} encodes a hard problem (assessing optimality) into the problem of finding a suitable entanglement-breaking channel---which
involves checking separability (of the Choi state), a famously strongly NP-hard problem \cite{Gurvits03,Ioannou07,Gharibian10}.
However, one could of course pick a specific entanglement-breaking $\Psi$ and get a simple sufficient criterion for optimality in the form of checking the kernel of a certain matrix; this is exactly what we will do. Beware that, much like how different witnesses detect different entangled states, there will not exist one such $\Psi$ which certifies the spanning property for all witnesses,
and some $\Psi$ will be worse than others (ideally, $\mathsf C(\Psi)$ should be on the boundary of the separable states).
However we already encountered an entanglement-breaking channel in Eq.~\eqref{eq:EB_example}, which as it turns out certifies optimality for almost all known witnesses, thus making for a very simple, yet quite powerful criterion for 
optimality. This is what the remainder of this subsection will be about.

First, let us explicitly apply the fact that entanglement-breaking maps send block-positive to positive semi-definite operators, to the channel from Eq.~\eqref{eq:EB_example}: doing so gives a first non-trivial, new necessary constraint on entanglement witnesses:
\begin{corollary}\label{coro_1}
    For all $W\in\mathbb C^{m\times m}\otimes\mathbb C^{n\times n}$ block-positive:
    \begin{align*}
    W+{\bf1}\otimes{\rm tr}_1(W)&\geq 0\\
    W+{\rm tr}_2(W)\otimes{\bf1}&\geq 0
    \end{align*}
\end{corollary}
\begin{proof}
We already know $({\rm id}\otimes\Psi)(W),(\Psi\otimes{\rm id})(W)\geq 0$ where $\Psi(X)=X+{\rm tr}(X){\bf1}$ is (a re-scaled version of the) map from Eq.~\eqref{eq:EB_example}. So because $$({\rm id}\otimes\Psi)(W)=W+({\rm id}\otimes{\rm tr}(\cdot){\bf1})(W)\,,$$ all we need to show is that
$({\rm id}\otimes{\rm tr}(\cdot){\bf1})(X)={\rm tr}_2(X)\otimes{\bf1}$ for all $X\in\mathbb C^{m\times m}\otimes\mathbb C^{n\times n}$ (the statement for $\Psi\otimes{\rm id}$ then is analogous).
Equivalently, by linearity it suffices to check this on all $X=A\otimes B$, i.e.,
\begin{align*}
    ({\rm id}\otimes{\rm tr}(\cdot){\bf1})(A\otimes B)={\rm tr}(B)A\otimes{\bf1}={\rm tr}_2(A\otimes B)\otimes{\bf1}
\end{align*}
which concludes the proof.
\end{proof}

Before moving on, some remarks are in order:

\begin{remark} \label{remark:operator_inequality_tightness}
\begin{itemize}
\item[(i)] At first glance Coro.~\ref{coro_1} resembles the reduction criterion which says that all separable states satisfy ${\rm tr}_2(\rho)\otimes {\bf1}-\rho,{\bf1}\otimes{\rm tr}_1(\rho)-\rho\geq 0$. However, the crucial differences to Coro.~\ref{coro_1} are that the sign is flipped and that the logic is reversed. To emphasize this, on the level of linear maps the reduction criterion follows from the fact that ${\rm tr}(\cdot){\bf1}-{\rm id}$ is positive, which is inequivalent to---and in fact not nearly as simple to show as---${\rm tr}(\cdot){\bf1}+{\rm id}$ being entanglement breaking.
    \item[(ii)] Coro.~\ref{coro_1} is tight in the sense that there exist non-zero witnesses for which $W+{\bf1}\otimes{\rm tr}_1(W)$, $W+{\rm tr}_2(W)\otimes{\bf1}$ have zero as an eigenvalue; more on this
    in a moment, cf.~also Coro.~\ref{cor:zero_eigval_to_weak_opt} below.
    \item[(iii)] Sometimes---for example, in the context of Bell inequalities---witnesses are defined with respect to some threshold $C$, that is, as operators $W$ which satisfy ${\rm tr}(W\rho)\geq C$ for all separable states $\rho$, but there exists some entangled state $\sigma$ such that ${\rm tr}(W\sigma)<C$. This is of course equivalent to the standard framework where $C=0$ using the simple shift $W\mapsto W-C\cdot{\bf1}$.
    Yet, introducing such a threshold can be beneficial for relating witnesses to local measurements \cite{Terhal00}. The reason we mention this is to illustrate how Coro.~\ref{coro_1} changes in this scenario: if $W$ is a witness \textnormal{to the threshold} $C$, so $W-C\cdot{\bf1}$ is still block-positive, then
    \begin{equation}\label{eq:witness_shift}
    \begin{split}
    W+{\bf1}\otimes{\rm tr}_1(W)&\geq C(n+1){\bf1}\\
    W+{\rm tr}_2(W)\otimes{\bf1}&\geq C(m+1){\bf1}\,.
    \end{split}
    \end{equation}
    Analogously, if a witness is defined with respect to an upper bound (i.e., ${\rm tr}(W\rho)\leq C$ for all $\rho$ separable) then~\eqref{eq:witness_shift} holds once $\geq$ is replaced by $\leq$.
\end{itemize}
\end{remark}
With this we are ready to state the sufficient optimality criterion which the depolarizing channel induces via Thm.~\ref{thm_main_new}~(iv):
\begin{theorem} \label{thm:full_Schmidt_rank_to_optimal}
    Given any witness $W\in\mathbb C^{m\times m}\otimes\mathbb C^{n\times n}$,
    the following statements hold.
    \begin{itemize}
        \item[(i)] If $m\leq n$ and if $\ker(W+{\rm tr}_2(W)\otimes{\bf1})$ contains a vector of Schmidt rank $m$, then $W$ is optimal.
        \item[(ii)] If $m\geq n$ and if $\ker(W+{\bf1}\otimes{\rm tr}_1(W))$ contains a vector of Schmidt rank $n$, then $W$ is optimal.
    \end{itemize}
\end{theorem}
\begin{proof}
(i): 
$\Psi:\mathbb C^{n\times n}\to\mathbb C^{n\times n}$, $\Psi
(X)=\frac{1}{n+1}(X+{\rm tr}(X){\bf1}_n)$ from Eq.~\eqref{eq:EB_example} is an entanglement-breaking channel, and it has full Choi rank because $\mathsf C(\Psi)=\frac{1}{n+1}(|\Gamma\rangle\langle\Gamma|+{\bf1})>0$.
Moreover, as seen in the proof of Coro.~\ref{coro_1}
$({\rm id}\otimes\Psi)(X)=\frac{1}{n+1}(X+{\rm tr}_2(X)\otimes{\bf1})$ for all $X$.
Hence we apply Thm.~\ref{thm_main_new}~(iv),(a) for $k=n\geq m$ to find that---because $\ker(({\rm id}\otimes\Psi)(W))=\ker(W+{\rm tr}_2(W)\otimes{\bf1})$ contains a vector of Schmidt rank $m$---$W$ has the spanning property. But as explained previously this is well-known to imply that $W$ is an optimal entanglement witness.
(ii): Repeat the argument from (i), now for $\Psi:\mathbb C^{m\times m}\to\mathbb C^{m\times m}$, $\Psi
(X)=\frac{1}{m+1}(X+{\rm tr}(X){\bf1}_m)$ and Thm.~\ref{thm_main_new}~(iv),(b) for $k=m\geq n$. Here one needs that $\Psi$ is self-adjoint which is readily verified.
\end{proof}
\noindent In line with Lemma~\ref{lemma_optimal_family}, this optimality criterion, if satisfied, transfers from one witness to the entire family $\{(X\otimes{\bf1})W(X^\dagger\otimes{\bf1}):X\in\mathbb C^{m\times m}\text{ full rank}\}$ of witnesses because the transformation $W\to(X\otimes{\bf1})W(X^\dagger\otimes{\bf1})$ results in $$W+{\rm tr}_2(W)\otimes{\bf1}\to (X\otimes{\bf1})(W+{\rm tr}_2(W)\otimes{\bf1})(X^\dagger\otimes{\bf1})\,.$$
In particular, the rank of any kernel element of \mbox{$W+{\rm tr}_2(W)\otimes{\bf1}$} is preserved under this transformation because $X$ has full rank.

\begin{remark}\label{rem_kernel_weaker_than_spanning}
\begin{itemize}
\item[(i)]
Because we chose a particular entanglement-breaking channel in Thm.~\ref{thm_main_new}~(iv), unsurprisingly, the corresponding optimality criterion (Thm.~\ref{thm:full_Schmidt_rank_to_optimal}) is 
weaker than the spanning property:
an example of a witness which has the spanning property but does not satisfy our Thm.~\ref{thm:full_Schmidt_rank_to_optimal} is the flip operator in \textnormal{odd} dimensions. For the simplest case of two qutrits, the spanning property of $\mathbb F$ is readily verified,
but any element of \mbox{$\ker(W_{{}^T}+{\rm tr}_2(W_{{}^T})\otimes{\bf1})=\ker(\mathbb F + {\bf 1})$} has the form $(0,a,b,-a,0,c,-b,-c,0)^T$ for some $a,b,c\in\mathbb C$, the Schmidt rank of which---by definition---is the rank of
$$
\begin{pmatrix}
    0&-a&-b\\a&0&-c\\b&c&0
\end{pmatrix}\,.
$$
But this matrix has determinant $0$ meaning its rank cannot be $3$. Therefore $\ker(W_{{}^T}+{\rm tr}_2(W_{{}^T})\otimes{\bf1})$ does not contain any vector of full Schmidt rank.
\item[(ii)] There is an interesting connection between Thm.~\ref{thm:full_Schmidt_rank_to_optimal} and the structural physical approximation (SPA) conjecture. Although the SPA conjecture was disproven \cite{CS14b}, it originally proposed---given a witness $W$---to study separability of \mbox{$({\rm id}\otimes\Psi_D)(W)$} with \mbox{$\Psi_D(X):=(1-p)X+p\,{\rm tr}(X)\frac{\bf1}{n}$} and $p=p(W)$ is chosen to be the smallest value for which \mbox{$({\rm id}\otimes\Psi_D)(W)\geq 0$} \cite{KAB08,CPS09,Bae17}; one always has $p(W)\leq\frac{n}{n+1}$ because $\Psi_D$ is entanglement-breaking for $p=\frac{n}{n+1}$, recall Sec.~\ref{sec_EW_pos}. Notably, SPA re-optimizes $p$ witness-by-witness, whereas we simply fix the universal upper bound $p=\frac{n}{n+1}$ when applying $\Psi_D$ because we, crucially, only require 
$({\rm id}\otimes\Psi_D)(W)$ to be positive semi-definite (not  separable).
In this framework, Thm.~\ref{thm:full_Schmidt_rank_to_optimal} states that if the minimal $p(W)$ for some witness $W$ is indeed $\frac{n}{n+1}$, and \textnormal{if} the corresponding kernel vector has maximal Schmidt rank, \textnormal{then} $W$ is optimal.
\end{itemize}
\end{remark}
Although this remark at first glance suggests otherwise, Thm.~\ref{thm:full_Schmidt_rank_to_optimal} turns out to be quite useful as it establishes optimality for many well-known witnesses.
To substantiate this, as a first example let us  
again look at the ${\sf SWAP}$ operator $\mathbb F$, but now for even dimensions.
Consider
$| \tilde{\Gamma} \rangle = \sum_{\ell = 0}^{n-1} (-1)^\ell |\ell\rangle \otimes |n - 1 - \ell \rangle$
which turns out to be 
a $-1$ eigenstate of $\mathbb F$:
\begin{align*}
    \mathbb F | \tilde{\Gamma} \rangle & 
    = \sum_{\ell = 0}^{n-1} (-1)^\ell \mathbb F\big(|\ell\rangle \otimes |n - 1 - \ell \rangle \big) \\ & 
    = \sum_{\ell = 0}^{n-1} (-1)^\ell |n - 1 - \ell \rangle \otimes |\ell\rangle \\ & 
    = \sum_{\ell' = 0}^{n-1} (-1)^{n - 1 - \ell'} |\ell' \rangle \otimes |n - 1 - \ell'\rangle \\ & 
    = (-1)^{n - 1} \sum_{\ell' = 0}^{n-1} (-1)^{\ell'} |\ell' \rangle \otimes |n - 1 - \ell'\rangle \\ & 
    = \begin{cases}
        | \tilde{\Gamma} \rangle & n \text{ odd} \\
        - | \tilde{\Gamma} \rangle & n \text{ even} \, .
    \end{cases}
\end{align*}
As a consequence, $| \tilde{\Gamma} \rangle$ is
an element of the kernel of \mbox{$W_{{}^T}+{\rm tr}_2(W_{{}^T})\otimes{\bf1}=\mathbb F + {\bf 1}$}. 
Because $| \tilde{\Gamma} \rangle$
has maximal Schmidt rank, 
Thm.~\ref{thm:full_Schmidt_rank_to_optimal} shows that the ${\sf SWAP}$ operator is an optimal entanglement witness whenever
$m=n$ is even.
Therefore, Lemma~\ref{lemma_optimal_family} shows that $(X\otimes{\bf1}){\mathbb F}(X^\dagger\otimes{\bf1})$ is also an optimal entanglement witness in this case for any full rank operator $X$.

The witness from the reduction map is shown to be optimal in a similar way.
Indeed, the bound from Coro.~\ref{coro_1} is reached by the (unnormalized) canonical maximally entangled state $|\Gamma \rangle$.
The corresponding inequality then reads $n\cdot{\bf1}-|\Gamma\rangle\langle\Gamma|\geq 0$ which is tight because applying $|\Gamma\rangle$ to it equals $0$.
The state $|\Gamma \rangle$ (full Schmidt rank) is thus in $\ker(W_{\rm red}+{\rm tr}_2(W_{\rm red})\otimes{\bf1})$
so 
\mbox{$W_{\mathrm{red}} = {\bf 1}_{n} \otimes {\bf 1}_{n} - | \Gamma \rangle \langle \Gamma |$} is optimal.

Before checking more witnesses we make the following observation:
in both previous examples \mbox{$\ker(W+{\rm tr}_2(W)\otimes{\bf1})$} contained not just a vector of full Schmidt rank, but even a maximally entangled state. It turns out that in this case one finds an even simpler sufficient criterion for optimality via expectation values \footnote{
A similar condition could be formulated for the general case of a kernel vector of full Schmidt rank ${\rm vec}(X)$; however, one would have to replace ${\rm tr}(W)$ in~\eqref{eq:coro_2_1} by something like ${\rm tr}(({\bf1}\otimes XX^\dagger)W)$, thus leading to a state-dependent bound which is not as nice to work with anymore. 
}:

\begin{corollary}\label{coro:trace_bound_from_max_entangled}
For all witnesses $W\in\mathbb C^{m\times m}\otimes\mathbb C^{n\times n}$ and all maximally entangled states $\Omega$---i.e., $\Omega$ has Schmidt rank $\min\{m,n\}$ and all its Schmidt coefficients equal $(\min\{m,n\})^{-1/2}$---it holds that
\begin{equation}\label{eq:coro_2_1}
    \langle\Omega|W|\Omega\rangle\geq-\,\frac{{\rm tr}(W)}{\min\{m,n\}}\,.
\end{equation}
Moreover, if equality holds in~\eqref{eq:coro_2_1} for some maximally entangled state $\Omega$, then $W$ is optimal.
\end{corollary}
\begin{proof}
Without loss of generality $m\leq n$. Because $\Omega$ is maximally entangled there exist an orthonormal basis $\{u_j\}_{j=1}^m$ of $\mathbb C^m$ and an orthonormal system $\{v_j\}_{j=1}^m$ in $\mathbb C^n$ such that $|\Omega\rangle=\sum_{j=1}^m\frac1{\sqrt m}|u_j\rangle\otimes |v_j\rangle$. This lets us compute
\begin{align*}
    & \langle\Omega|W+{\rm tr}_2(W)\otimes{\bf1}|\Omega\rangle\\
    &\quad =\langle\Omega|W|\Omega\rangle+\frac1m\sum_{j,k}\langle u_j\otimes v_j|{\rm tr}_2(W)\otimes{\bf1}|u_k\otimes v_k\rangle\\
    &\quad =\langle\Omega|W|\Omega\rangle+\frac1m\sum_{j}\langle u_j|{\rm tr}_2(W)|u_j\rangle\\
    &\quad =\langle\Omega|W|\Omega\rangle+\frac{{\rm tr}({\rm tr}_2(W))}{m}=\langle\Omega|W|\Omega\rangle+\frac{{\rm tr}(W)}{m}\,.
\end{align*}
By Coro.~\ref{coro_1} this expression is non-negative so~\eqref{eq:coro_2_1} holds. Also if equality holds in~\eqref{eq:coro_2_1}, then this shows that \mbox{$\langle\Omega|W+{\rm tr}_2(W)\otimes{\bf1}|\Omega\rangle=0$}. Thus---again because \mbox{$W+{\rm tr}_2(W)\otimes{\bf1}\geq 0$} by Coro.~\ref{coro_1}---this implies \mbox{$\Omega\in\ker(W+{\rm tr}_2(W)\otimes{\bf1})$}. Hence $W$ is optimal by Thm.~\ref{thm:full_Schmidt_rank_to_optimal}~(i). The case $m\geq n$ is proven analogously.
\end{proof}
\noindent The bound from Coro.~\ref{coro:trace_bound_from_max_entangled} is tight and, in particular, equality in~\eqref{eq:coro_2_1} for some state $\Omega$ necessitates that $-\,\frac{{\rm tr}(W)}{\min\{m,n\}}$ is an eigenvalue of $W$. Once that is certified, it suffices to check whether the corresponding eigenvector (resp.~whether some eigenvector from the corresponding eigenspace) is maximally entangled.

Now Coro.~\ref{coro:trace_bound_from_max_entangled} can be used to easily show optimality of not just the transpose in even dimensions or the reduction witness (cf.~also~\eqref{eq:proof_of_thm2_1} below), but more generally of the Breuer-Hall witness
$W_{\rm BH} \in \mathbb C^{2n\times 2n} \otimes \mathbb C^{2n\times 2n}$ (recall Eq.~\eqref{eq:BreuerHallWitness})
\begin{equation*}
    W_{\rm BH} = {\bf 1} \otimes {\bf 1} - | \Gamma \rangle \langle \Gamma | - ({\bf1}\otimes U) \mathbb F ({\bf1}\otimes U^\dagger)
\end{equation*}
with $U$ sub-unitary and antisymmetric.
Indeed
\begin{align*}
    {\rm tr}(W_{\rm BH}) & 
    = (2n)^2 - \langle \Gamma | \Gamma \rangle - {\rm tr}(({\bf 1} \otimes U^\dagger U)\mathbb F) \\ & 
    = (2n)^2 - 2n - {\rm tr}(U^\dagger U)\, ,
\end{align*}
resulting in the bound 
$\langle \Omega | W _{\rm BH}| \Omega \rangle \geq 1 + \frac{1}{2n}{\rm tr}(U^\dagger U) - 2n $ 
for any normalized maximally entangled state (Eq.~\eqref{eq:coro_2_1}). 
Expanding the left-hand side
\begin{equation*}
    \langle \Omega | W_{\rm BH} | \Omega \rangle = \langle \Omega | \Omega \rangle - |\langle \Omega | \Gamma \rangle|^2 - \langle \Omega | ({\bf 1} \otimes U) \mathbb F ( {\bf 1} \otimes U^\dagger) | \Omega \rangle \, ,
\end{equation*}
one sees that choosing $| \Omega \rangle = \frac{1}{\sqrt{2n}} | \Gamma \rangle$ results in $\langle \Omega | \Omega \rangle = 1$, $|\langle \Omega | \Gamma \rangle|^2 = 2n$, and 
\begin{equation*}
    \langle \Omega | ({\bf 1} \otimes U) \mathbb F ( {\bf 1} \otimes U^\dagger) | \Omega \rangle 
    = \frac{1}{2n} {\rm tr}(\overline{U} U) 
    = - \frac{1}{2n} {\rm tr}(U^\dagger U) 
\end{equation*}
where in the last step we used that $U$ is antisymmetric. 
Putting everything together, 
\begin{equation*}
    \langle \Omega | W_{\rm BH} | \Omega \rangle = 1 + \frac{1}{2n} {\rm tr}(U^\dagger U) - 2n = -\frac{{\rm tr}(W_{\rm BH})}{2n} \, ,
\end{equation*}
showing that $W_{\rm BH}$ is optimal (by Coro.~\ref{coro:trace_bound_from_max_entangled}).
Beware that optimality of the Choi map (resp.~the corresponding witness~\eqref{eq:Choi_witness}) is not covered by our criteria---and in fact it cannot be covered by any criterion resulting from Thm.~\ref{thm_main_new}---because this witness does not have the spanning property (recall Sec.~\ref{sec_EW_pos}).

\begin{remark}\label{rem_fullyentangledfraction}
At first glance, Eq.~\eqref{eq:coro_2_1} resembles the criterion $\max_\Omega\langle\Omega|A|\Omega\rangle>\frac{1}{\min\{m,n\}}{\rm tr}(A)$ which guarantees that a bipartite operator $A\in\mathbb C^{m\times m}\otimes\mathbb C^{n\times n}$, $A\geq 0$ is entangled \cite{HH99}
(this maximum is also known as ``fully entangled fraction'' \cite{BDSW96}).
However, the fundamental difference between these criteria the same as in Rem.~\ref{remark:operator_inequality_tightness}~(i), because they are derived from two fundamentally different inequalities: our witness operator inequality (Coro.~\ref{coro_1}) and the reduction criterion, respectively. 
\end{remark}

We conclude this section with a simple sufficient (but not necessary) criterion for weak optimality of entanglement witnesses; this follows at once from Rem.~\ref{rem_weak_opt} when applied to the entanglement-breaking channel from Eq.~\eqref{eq:EB_example}:
\begin{corollary} \label{cor:zero_eigval_to_weak_opt}
Given any witness $W\in\mathbb C^{m\times m}\otimes\mathbb C^{n\times n}$, if $0$ is an eigenvalue of $W+{\rm tr}_2(W)\otimes{\bf1}$ or $W+{\bf1}\otimes{\rm tr}_1(W)$, then $W$ is weakly optimal.
\end{corollary}

\subsection{Implications for positive maps and spectra of witnesses}\label{sec_pos}

An obvious question at this point is whether these results on entanglement witnesses lead---via CJ%
---to any new insights on positive maps. 
To answer this let us first translate Coro.~\ref{coro_1}: given $\Phi:\mathbb C^{n\times n}\to\mathbb C^{m\times m}$ linear and positive, setting $W=\mathsf C(\Phi^\dagger)$ and using that ${\rm tr}_1(\mathsf C(\Phi^\dagger))=\Phi^\dagger({\bf1})$ and ${\rm tr}_2(\mathsf C(\Phi^\dagger))=\Phi({\bf1})^T$ (as is readily verified) one finds
    \begin{align*}
    \mathsf C(\Phi^\dagger)+{\bf1}\otimes\Phi^\dagger({\bf1})&\geq 0\\
    \mathsf C(\Phi^\dagger)+\Phi({\bf1})^T\otimes{\bf1}&\geq 0\,.
    \end{align*}
These are new spectral constraints for Choi matrices of positive maps (more on this also at the end of this section). Notably, this shows that for positive maps which are additionally trace-preserving or unital, all eigenvalues of the corresponding Choi matrix are lower bounded by $-1$.

Next, Coro.~\ref{coro:trace_bound_from_max_entangled} places a lower bound on $\langle\Omega|\mathsf C(\Phi^\dagger)|\Omega\rangle$ which---if $m=n$---relates to (a version of) the superoperator trace ${\rm tr}(\Phi)=\sum_{j,k=1}^n\langle j|\Phi(|j\rangle\langle k|)|k\rangle$ of $\Phi$ \cite[Eq.~(6)]{vE24_decomp_findim}.
Notably, this trace connects to the entanglement fidelity $F_e(\Phi,\rho)$ \cite{Schumacher96}, \cite[Def.~3.30]{Watrous18} in that $F_e(\Phi,\rho)={\rm tr}(\Phi(\rho(\cdot)\rho))$ for all $\Phi$ completely positive and all $\rho\geq 0$ (if $\rho$ is the maximally mixed state then this reduces to the channel fidelity $\frac1{n^2}{\rm tr}(\Phi)$ \cite{DKSW07}). 
The known bound for this quantity from spectral considerations is ${\rm tr}(\Phi)\in[-n^2\|\Phi^\dagger({\bf1})\|_\infty,n^2\|\Phi^\dagger({\bf1})\|_\infty]$ \cite[Coro.~2.3.8]{Bhatia07}, where the upper bound is tight. The lower bound, however, can be strengthened further using our new results:
\begin{theorem}\label{thm_2}
    Let $\Phi:\mathbb C^{n\times n}\to\mathbb C^{n\times n}$, $n\in\mathbb N$ be linear and positive. Then the following statements hold:
    \begin{itemize}
        \item[(i)] One has ${\rm tr}(\Phi)\geq -{\rm tr}(\Phi({\bf1}))$ or, equivalently, ${\rm tr}(\Phi)\geq -{\rm tr}(\Phi^\dagger({\bf1}))$.
    Moreover, this bound is tight, i.e., for all $n\in\mathbb N$ there exists $\Phi\neq 0$
    positive such that ${\rm tr}(\Phi)= -{\rm tr}(\Phi({\bf1}))$.
    \item[(ii)] $
        {\rm tr}(\Phi)\geq-n\min\{\|\Phi({\bf1})\|_\infty,\|\Phi^\dagger({\bf1})\|_\infty\}
        $
    \item[(iii)] If $\Phi$ is additionally trace-preserving or unital, then ${\rm tr}(\Phi)\geq -n$.
    \item[(iv)] If there exists $U\in\mathbb C^{n\times n}$ unitary such that ${\rm tr}(\Phi(U^\dagger(\cdot)U))= -{\rm tr}(\Phi({\bf1}))$, then $\mathsf C(\Phi)$ and $\mathsf C(\Phi^\dagger)$ are optimal entanglement witnesses.
    In particular, this conclusion can be drawn if ${\rm tr}(\Phi)= -{\rm tr}(\Phi({\bf1}))$.
    \end{itemize}
\end{theorem}
\begin{proof}
(i): Because $\Phi$ is positive, $\mathsf C(\Phi)$ is block-positive so Coro.~\ref{coro:trace_bound_from_max_entangled} implies
$\frac1n\langle\Gamma|\mathsf C(\Phi)|\Gamma\rangle\geq -\frac1n{\rm tr}(\mathsf C(\Phi))$. Now $\langle\Gamma|\mathsf C(\Phi)|\Gamma\rangle={\rm tr}(\Phi)$ \cite[Lemma~2]{vE24_decomp_findim} as well as \mbox{${\rm tr}(\mathsf C(\Phi))={\rm tr}(\Phi({\bf1}))={\rm tr}(\Phi^\dagger({\bf1}))$} which combines to the desired inequalities.
Finally, equality is achieved, e.g., by the reduction map $R={\rm tr}(\cdot){\bf1}-(\cdot)$:
    because \mbox{$R({\bf1})=(n-1){\bf1}$} one has ${\rm tr}(R({\bf1}))=n^2-n$ and hence
    \begin{align}
        {\rm tr}(R)&=\sum_{j,k=1}^n\big\langle j\big| \big( {\rm tr}(|j\rangle\langle k|){\bf1}- |j\rangle\langle k|\big) \big|k\big\rangle\notag\\
        &=n-n^2=-{\rm tr}(R({\bf1}))\,.\label{eq:proof_of_thm2_1}
    \end{align}
    
(ii): Apply the inequality \mbox{$X\leq\|X\|_\infty{\bf1}$} (which is well known to hold for all $X\in\mathbb C^{n\times n}$ Hermitian) to $\Phi({\bf1}),\Phi^\dagger({\bf1})$ and combine this with~(i).

(iii): Follows from (i) or (ii) using that $\Phi$ is trace-preserving (unital) if and only if $\Phi^\dagger({\bf1})={\bf1}$ ($\Phi({\bf1})={\bf1}$).

(iv): As before the key here is \cite[Lemma~2]{vE24_decomp_findim} which states that ${\rm tr}(\Phi(U^\dagger(\cdot)U))=\langle{\rm vec}(U)|\mathsf C(\Phi)|{\rm vec}(U)\rangle$ for all $U$.
Now if $U$ is unitary, then $|\Omega\rangle:=\frac1{\sqrt n}|{\rm vec}(U)\rangle=\frac1{\sqrt n}({\bf1}\otimes U)|\Gamma\rangle$ is maximally entangled so we know that
\begin{align*}
    \langle\Omega|\mathsf C(\Phi)|\Omega\rangle&=\frac1n{\rm tr}(\Phi(U^\dagger(\cdot)U))=\frac1n{\rm tr}(U^\dagger\Phi(\cdot)U)\\
    &=-\frac1n{\rm tr}(\Phi({\bf1}))=-\frac1n{\rm tr}(\mathsf C(\Phi))\,;
\end{align*}
in the second step we used that the trace is cyclic.
Hence Coro.~\ref{coro:trace_bound_from_max_entangled} guarantees that $\mathsf C(\Phi)$ is optimal. For $\mathsf C(\Phi^\dagger)$ use ${\rm tr}(\Phi)={\rm tr}(\Phi^\dagger)$ (true for all $\Phi$ Hermitian-preserving) so, defining $|\Omega\rangle:=\frac1{\sqrt n}({\bf1}\otimes U^\dagger)|\Gamma\rangle$, one similarly finds
\begin{align*}
    \langle\Omega|\mathsf C(\Phi^\dagger)|\Omega\rangle&=\frac1n{\rm tr}(U\Phi^\dagger(\cdot)U^\dagger)=\frac1n{\rm tr}(\Phi(U^\dagger(\cdot)U))\\
    &=-\frac1n{\rm tr}(\Phi({\bf1}))=-\frac1n{\rm tr}(\Phi^\dagger({\bf1}))
\end{align*}
which is in turn equal to $-\frac1n{\rm tr}(\mathsf C(\Phi^\dagger))$. Hence Coro.~\ref{coro:trace_bound_from_max_entangled} again guarantees optimality of $\mathsf C(\Phi^\dagger)$.
In particular, setting $U={\bf1}$ implies the final statement.
\end{proof}
\begin{remark}\label{rem_optimize_witness_functional}
For the purpose of numerics, what these results imply is that given some $\Phi:\mathbb C^{n\times n}\to\mathbb C^{n\times n}$ positive---resp.~given the corresponding witness $W=\mathsf C(\Phi^\dagger)$---one can implement the optimization problem
\begin{equation}\label{eq:cost_witnessopt}
    \min_U\langle\Gamma|({\bf1}\otimes U^\dagger)W({\bf1}\otimes U)|\Gamma\rangle
\end{equation}
over the $n$-dimensional unitary group, and if this minimum equals $-\frac{1}{n^2}{\rm tr}(W)$, then the witness is optimal.
While this is not sufficient for optimality, this, to our knowledge, is the first functional which can be used to ascertain optimality.
Moreover, \eqref{eq:cost_witnessopt}---or perhaps just $W\mapsto{\rm tr}(W)+n\langle\Gamma|W|\Gamma\rangle$ (Coro.~\ref{coro:trace_bound_from_max_entangled})---may serve as a cost function to optimize a given witness $W$ towards becoming (``more'') optimal, assuming one can suitably encode block-positivity of $W$ into such an optimization.
\end{remark}

With this in mind let us revisit the positive maps behind the witnesses we looked at in Sec.~\ref{sec_opt}. The transpose always leads to optimal witnesses, but satisfies the trace condition from Thm.~\ref{thm_2}~(iv) only for $n$ even (which is in line with Remark~\ref{rem_kernel_weaker_than_spanning}~(i) where we saw that in odd dimensions, our sufficient criterion fails for the transpose):
${\rm tr}(U^\dagger(\cdot)^TU)={\rm tr}(\overline UU)$,
the minimum of which is either $-n$ if $n$ is even, or $-(n-2)$ if $n$ is odd, cf.~Lemma~\ref{lemma_appA_1} in Appendix~\ref{appB}.

The re-scaled reduction map (by $\frac1{n-1}$) as well as the re-scaled Breuer-Hall maps (by $\frac1{2n-2}$, with $U$ \textit{unitary} and anti-symmetric) are positive and trace-preserving and they both have trace $-n$, meaning Thm.~\ref{thm_2}~(iv) implies optimality of the associated witnesses.
This optimality is also reflected in the region of positive, trace-preserving maps from \cite[Eq.(72) \& Lemma 24]{KMS20} which also takes the minimal value $-n$ (for $\alpha=\beta=-\frac{1}{n-1}$ in which case $M=\frac1{n-1}R$ with $R$ the reduction map from above), regardless of the dimension $n$.

However, there is one map which we still have to show optimality for: the generalizations of the Robertson map (\eqref{eq:gen_Rob_1} and \eqref{eq:gen_Rob_2}). It turns out that this is most easily done on the level of the map itself, using the criterion from our previous theorem (yet, we outsourced the proof to Appendix~\ref{appC} because it features some brute-force computation which is not very illuminating):

\begin{corollary}\label{coro_robertson}
    Let $\Phi_{2n}:\mathbb C^{2n\times 2n}\to \mathbb C^{2n\times 2n}$ 
    which maps any input $X$ to 
    \begin{equation}
    \label{eq:gen_Rob}
    \frac1n \begin{pmatrix}
        {\rm tr}(X_{22}){\bf1} & -X_{12}-R(X_{21})\\
        -X_{21}-R(X_{12}) & {\rm tr}(X_{11}){\bf1}
    \end{pmatrix}
    \end{equation}
    be given, where $X_{ij}$ are the corresponding blocks of $X$,
    and the linear map $R: \mathbb C^{n \times n} \to \mathbb C^{n \times n}$ is chosen such that $\Phi_{2n}$ is positive. Then the entanglement witness corresponding to $\Phi_{2n}$ is optimal.
\end{corollary}

The generalizations \eqref{eq:gen_Rob_1} and \eqref{eq:gen_Rob_2} of the Robertson map are covered by this corollary when choosing $R(X) = {\rm tr}(X){\bf 1} - X$ and $R(X)=UX^TU^\dagger$, respectively.
Note that the second map is only defined if $n$ itself is even, because antisymmetric unitary matrices only exist in even dimensions (consequence of Lemma~\ref{lemma_appA_1}, Appendix~\ref{appB}). Of course one can relax this to antisymmetric sub-unitary matrices (as done for the generalized Breuer-Hall map) and still get optimality by the previous corollary.

Finally, given that Coro.~\ref{coro:trace_bound_from_max_entangled} is essentially a spectral condition on block-positive matrices (since the trace is the sum of all eigenvalues), and Coro.~\ref{coro_1} imposes similar restrictions, let us quickly recap some known spectral constraints on block-positive operators $W\in\mathbb C^{m\times m}\otimes\mathbb C^{n\times n}$ \cite[Sec.~2.3]{JP18} to see how our results compare:
$W$ has no more than $(m-1)(n-1)$ negative eigenvalues (even in the decomposable case), one has ${\rm tr}(W)^2\geq{\rm tr}(W^2)$, and if $W$ is decomposable then the smallest eigenvalue of $W$ satisfies $\lambda_{\rm min}(W)\geq-\frac12{\rm tr}(W)$.
Moreover, there are some lower bounds on the ratio $\frac{\lambda_{\rm min}(W)}{\lambda_{\rm max}(W)}$ between the smallest and the largest eigenvalue of $W$ \cite[Coro.~5.5]{JK10} (which we will not spell out explicitly here).
To add to these constraints, our results imply the following new bounds:
\begin{corollary}\label{coro_wit_lambda}
    Let $W\in\mathbb C^{m\times m}\otimes\mathbb C^{n\times n}$ be any entanglement witness. The following statements hold:
    \begin{itemize}
        \item[(i)] $\lambda_{\rm min}(W)\geq -{\rm tr}(W)$
        \item[(ii)] All eigenvalues of $W$ are lower bounded by $-\lambda_{\rm max}({\rm tr}_1(W))$ as well as $-\lambda_{\rm max}({\rm tr}_2(W))$.
    \end{itemize}
\end{corollary}
\begin{proof}
The inequality both these results are based on is the following: because $W$ is Hermitian
\begin{align}
\lambda_{\rm min}(W)&=\min_{\|x\|=1}\langle x|W|x\rangle\notag\\
&\geq \min_{\|x\|=1}-\langle x| {\bf 1}\otimes{\rm tr}_1(W) |x\rangle\notag\\
&=-\max_{\|x\|=1}\langle x| {\bf 1}\otimes{\rm tr}_1(W) |x\rangle\label{eq:lowerbound_ev}
\end{align}
where in the second step we used Coro.~\ref{coro_1}.
(i): The idea here is that $ {\rm tr}_1(|x\rangle\langle x|) $ is a state, so it can be upper bounded by ${\bf1}$:
\begin{align*}
\langle x|{\bf 1}\otimes{\rm tr}_1(W)|x\rangle
&={\rm tr}\big( {\rm tr}_1(|x\rangle\langle x|){\rm tr}_1(W) \big)\\
&={\rm tr}\big(({\bf1}\otimes  {\rm tr}_1(|x\rangle\langle x|) )W\big)\leq{\rm tr}(W)\,.
\end{align*}
Combining this with~\eqref{eq:lowerbound_ev}, because of the minus sign, yields (i).
(ii): On the other hand, we could upper bound the expectation value in~\eqref{eq:lowerbound_ev} via the operator norm: 
    \begin{align*}
\langle x|{\bf 1}\otimes{\rm tr}_1(W)|x\rangle
        &\leq \|{\bf 1}\otimes{\rm tr}_1(W)\|_\infty\\
        &=\|{\rm tr}_1(W)\|_\infty=\lambda_{\rm max}({\rm tr}_1(W))
    \end{align*}
   In the last step we used that ${\rm tr}_1(W)\geq 0$.
   For ${\rm tr}_2(W)$ one argues analogously which, altogether, shows (ii).
\end{proof}
Two remarks on these bounds: First, we do not know of an example which shows that $\lambda_{\rm min}(W)\geq -{\rm tr}(W)$ is tight. Indeed, this is just what our results immediately imply, and it could even be that the $-\frac12{\rm tr}(W)$ bound from the decomposable case \cite{Rana13} holds for all witnesses.
Moreover, it is unclear to us how our bound involving the largest eigenvalue of the partial trace of $W$ compares to the known bounds on the ratios between smallest and largest eigenvalue of $W$. The problem here is that the effect of the the partial trace on the largest eigenvalue seems to not exhibit any ``nice'' behavior, beyond trivial bounds like $\lambda_{\rm max}({\rm tr}_1(X))\leq m\|X\|_\infty$.

\section{Conclusions}\label{sec_concl}

Using a new characterization of the spanning property via entanglement-breaking channels, in this paper we found two new, simple sufficient criteria for optimality of bipartite entanglement witnesses: First, if one of the two positive semi-definite operators $W+{\rm tr}_2(W)\otimes{\bf1}$, $W+{\bf1}\otimes{\rm tr}_1(W)$ has a vector of full Schmidt rank in its kernel, and second, if $\langle\Omega|W|\Omega\rangle\geq-\frac{{\rm tr}(W)}{\min\{m,n\}}$ is saturated for some maximally entangled state $|\Omega\rangle$.
In particular, for the purpose of numerics the latter can be turned into a straightforward minimization problem over the unitary group of one of the subsystems.
Moreover, for weak optimality it is sufficient that either $W+{\rm tr}_2(W)\otimes{\bf1}$ or $W+{\bf1}\otimes{\rm tr}_1(W)$ has a zero eigenvalue.

Finally, we translated our findings from witnesses to positive maps to find the following new lower bound on the superoperator trace of a positive map: \mbox{${\rm tr}(\Phi)\geq -{\rm tr}(\Phi({\bf1}))$}. Again, optimality of the associated witness is guaranteed if equality holds.
All these results were interleaved with numerous examples to show how simple, yet versatile our new criteria are for establishing optimality of entanglement witnesses.
Indeed, the only witness with the spanning property we found which our criteria could not identify as optimal was the flip operator in odd dimensions.

This leads to a handful of interesting follow-up questions:
First, is there a ``nice'' entanglement-breaking channel of full Choi rank which ``detects'' optimality of the flip in odd dimensions, i.e., for which $\ker((\Psi\otimes{\rm id})(\mathbb F))$ contains a vector of full Schmidt rank? 
While Thm.~\ref{thm_main_new} guarantees that such a channel \textit{exists}, the question is whether there exists one which can be written down as easily as $X\mapsto X+{\rm tr}(X){\bf1}$.
If so, this may lead to the discovery of new useful classes of entanglement-breaking channels.
Next: is our trace-based criterion (Coro.~\ref{coro:trace_bound_from_max_entangled}) actually equivalent to our kernel criterion (Thm.~\ref{thm:full_Schmidt_rank_to_optimal})?
This is the question mark in Fig.~\ref{fig:overview_diagram}, and while the kernel criterion seems stronger we were not able to find a witness for which only the kernel criterion holds.
Finally, is there any way to generalize our criteria to the multipartite setting? After all, the spanning property is sufficient for optimality for all finite-dimensional systems (the original proof of Lewenstein et al.~\cite{LKCH00} does not rely on there only being two subsystems), yet as soon as we invoked Choi-Jamio\l{}kowski, 
the Schmidt rank, or maximally entangled states we made use of inherently bipartite formalisms.
Thus, while an object like $W+{\bf1}\otimes\ldots\otimes{\bf1}\otimes{\rm tr}_{\neg j}(W)\otimes{\bf1}\otimes\ldots\otimes{\bf1}$ is of course well-defined, it is not even clear whether it is always positive semi-definite whenever $W$ is block-positive, let alone whether its kernel relates to optimality of $W$ in any way. 
While these questions promise valuable insights, their resolution remains a subject for future research.

\appendix
\section{Example: from spanning property to witness-annihilating full-rank separable state}\label{appA}

To better illustrate the first part of the proof of Thm.~\ref{thm_main_new} let us look at an explicit example to see how---starting from the spanning property---one can construct a full-rank separable state on which a given entanglement witness vanishes.
For this we consider the one-qubit transpose map $\Phi={}^T$ so $W=\mathbb F$.
For any orthogonal states $\{ |b_0 \rangle, |b_1 \rangle \} $ one has $\langle b_0 \otimes  b_1| \mathbb F |b_0 \otimes b_1 \rangle=|\langle b_0|b_1\rangle|^2 =0$. 
In particular, this holds for $|0\rangle$ and $|1\rangle$, as well as for $|\pm\rangle := \frac{1}{\sqrt{2}}|0\rangle \pm |1\rangle$ and $|R/L\rangle := \frac{1}{\sqrt{2}}|0\rangle \pm i |1\rangle$. 
Then the set 
$\{|0 1 \rangle, |1 0 \rangle, |\text{+ -}\rangle, |R L \rangle\} = \{(0, 1, 0, 0)^T, (0, 0, 1, 0)^T, \frac{1}{2}(1, -1, 1, -1)^T, \frac{1}{2}(1, -i, i, 1)^T\}$ 
is linearly independent, i.e., $4$-dimensional, and thus $\mathbb F\in\mathbb C^{4\times 4}$ has the spanning property.
From this, one can construct the rank-4 (i.e., full-rank) separable state
\begin{multline*}
    \frac{1}{4} ( |01\rangle\langle 01| + |10\rangle\langle 10| +  |\text{+ -}\rangle\langle \text{+ -}| +  |RL\rangle\langle RL| ) \\
    = \frac{1}{16}\begin{pmatrix}
        2 & -1+i & 1-i & 0 \\
        -1-i & 6 & -2 & 1-i \\
        1+i & -2 & 6 & -1+i \\
        0 & 1+i & -1-i & 2
    \end{pmatrix} =: \rho 
\end{multline*}
which satisfies ${\rm tr}(\rho W)={\rm tr}(\rho \mathbb F) = 0$, since the witness vanishes on each of the terms that makes up $\rho$, individually. 
Indeed, by this argument we could have taken any mixture of these four pure product states (and not just the canonical one, as also done in the proof of Thm.~\ref{thm_main_new}, which was mostly convenience): In fact, we have that ${\rm tr}(\rho_p\mathbb F)=0$ for all probability vectors $p\in\mathbb R^4$ where $\rho_p:=p_1|01\rangle\langle 01| + p_2|10\rangle\langle 10| +  p_3|\text{+ -}\rangle\langle \text{+ -}| +  p_4|RL\rangle\langle RL|$, and $\rho_p$ has full rank if and only if all entries of $p$ are positive.
In other words this construction turns the spanning property into an entire face of the separable states which the witness functional maps to zero.

\section{A bound on ${\rm tr}(\overline UU)$}\label{appB}

\begin{lemma}\label{lemma_appA_1}
    For all $n\in\mathbb N$ and all $U\in\mathsf U(n)$ it holds that $|{\rm tr}(\overline{U}U)|\leq n$.
    Moreover,
    \begin{equation}
        \min_{U\in\mathsf U(n)}{\rm tr}(\overline{U}U)=\begin{cases}
        -n&n\text{ even}\\-(n-2)&n\text{ odd}\,.
    \end{cases}\label{eq:lemma_appA_1_1}
    \end{equation}
\end{lemma}
\begin{proof}
    First, the general bound. By Cauchy-Schwarz
    \begin{align*}
        |{\rm tr}(\overline{U}U)|&=|\langle U^T,U\rangle_{\rm HS}|\\
        &\leq\|U^T\|_2\|U\|_2=(\sqrt n)^2=n\,.
    \end{align*}
    If $n$ is even, the lower bound can be achieved: because $\frac n2\in\mathbb N$, ${\bf1}_{\frac n2}\otimes\sigma_y\in\mathsf U(n)$ is well-defined and satisfies
    $$
    {\rm tr}(({\bf1}_{\frac n2}\otimes\overline{\sigma_y})({\bf1}_{\frac n2}\otimes\sigma_y))={\rm tr}({\bf1}_{\frac n2}){\rm tr}(\overline{\sigma_y}\sigma_y)
    =-n\,.
    $$
    For odd $n$ the bound in~\eqref{eq:lemma_appA_1_1} is achieved, e.g., by
$$
\begin{pmatrix}
{\bf1}_{\frac{n-1}2}\otimes\sigma_y&0\\0&1
\end{pmatrix}
$$
as a straightforward computation shows.
To see that one cannot do better in the case where $n$ is odd we make the following observations:
\begin{itemize}
    \item $\overline UU$ is unitary, so all its eigenvalues satisfy $|\lambda_j|=1$
    \item $\det(\overline UU)=|\det U|^2=1$ because $U$ is unitary
    \item The eigenvalues of $\overline UU$ always come in complex conjugate pairs: $\overline UU$ and $U\overline U$ have the same eigenvalues, but---taking the entrywise conjugate---one sees that the eigenvalues of $\overline{\overline UU}=U\overline U$ are the complex conjugates of the eigenvalues of $U\overline U$\medskip
\end{itemize}
We use these facts to show that one of the eigenvalues of $\overline{U}U$ has to be $1$. Group all the eigenvalues of $\overline UU$ into two vectors: $v_0\in\mathbb C^m$, $m\geq 0$ contains all eigenvalues of $\overline UU$ with non-trivial imaginary part, and $v_1\in\{-1,1\}^{n-m}$ contains all other eigenvalues. As seen before the entries of $v_0$ come in complex conjugate pairs, which has two consequences: on the one hand, $m$ must be even---hence $n-m$ is odd---and on the other hand $\prod_{j=1}^{m}(v_0)_j=1$ (because $|(v_0)_j|^2=1$ for all $j$). Therefore
$
1=\det(\overline UU)=\prod_{j=1}^{m}(v_0)_j\cdot\prod_{j=1}^{n-m}(v_1)_j=\prod_{j=1}^{n-m}(v_1)_j$.
As $(v_1)_j=\pm1$ by construction, this can only be true if the number of negative entries of $v_1$ is even.
But $v_1$ has an odd number of entries, so we can conclude that at least one of its entries, i.e., one eigenvalue of $\overline UU$ has to be $1$.
Altogether, this lets us conclude the proof:
\begin{align}
    {\rm tr}(\overline UU)&=1+\sum_{j=2}^{n}\lambda_j(\overline UU)\notag\\
    &\geq 1-\Big|\sum_{j=2}^{n}\lambda_j(\overline UU)\Big|\notag\\
    &\geq 1-\sum_{j=2}^n|\lambda_j(\overline UU)|\notag\\
    &= 1-(n-1)=-(n-2)
    \tag*{\qedhere}
\end{align}
\end{proof}
\section{Proof of Coro.~\ref{coro_robertson}}\label{appC}
\begin{proof}
    By Thm.~\ref{thm_2}~(iv) it suffices to show ${\rm tr}(\Phi) = -{\rm tr}(\Phi({\bf1}))$.
    In the case of maps of the form~\eqref{eq:gen_Rob}, 
    $
        {\rm tr}(\Phi_{2n}({\bf1})) = {\rm tr}({\bf1}_{2n}) = 2n \, .
    $
    On the other hand, we have 
    \begin{align}
        {\rm tr}(\Phi_{2n}) & 
        = \sum_{j, k=1}^{2n} \langle j | \Phi_{2n}(|j \rangle \langle k|) |k\rangle \nonumber \\ & 
        = \sum_{j, k=1}^{2n} {\rm tr} \left((| j \rangle \langle k |)^T \Phi_{2n}(|j \rangle \langle k|) \right) \label{eq:trace_gen_Rob} \, .
    \end{align}
    The sum over $j$ and $k$ can be decomposed into four components, depending on the quadrant in which the matrix unit $| k \rangle \langle j |$ containing the $1$ is:
    \begin{widetext}
    \begin{align*}
        {\rm tr}(\Phi_{2n}) 
        = & \sum_{j, k=1}^{2n} {\rm tr} \left((| j \rangle \langle k |)^T \Phi_{2n}(|j \rangle \langle k|) \right) \\ 
        \begin{split}
        = & \sum_{\zeta, \kappa=1}^{n} {\rm tr} \left(\begin{pmatrix}
            |\kappa \rangle \langle \zeta| & 0 \\ 0 & 0
        \end{pmatrix} 
        \Phi_{2n} \begin{pmatrix}
            |\zeta \rangle \langle \kappa| & 0 \\ 0 & 0
        \end{pmatrix} \right) 
        + \sum_{\zeta, \kappa=1}^{n} {\rm tr} \left(\begin{pmatrix}
            0 & 0 \\ |\kappa \rangle \langle \zeta| & 0
        \end{pmatrix} 
        \Phi_{2n} \begin{pmatrix}
            0 & |\zeta \rangle \langle \kappa| \\ 0 & 0
        \end{pmatrix} \right) \\ 
        & \qquad \qquad \qquad + \sum_{\zeta, \kappa=1}^{n} {\rm tr} \left(\begin{pmatrix}
            0 & |\kappa \rangle \langle \zeta| \\ 0 & 0
        \end{pmatrix} 
        \Phi_{2n} \begin{pmatrix}
            0 & 0 \\ |\zeta \rangle \langle \kappa| & 0
        \end{pmatrix} \right)
        + \sum_{\zeta, \kappa=1}^{n} {\rm tr} \left(\begin{pmatrix}
            0 & 0 \\ 0 & |\kappa \rangle \langle \zeta|
        \end{pmatrix} 
        \Phi_{2n} \begin{pmatrix}
            0 & 0 \\ 0 & |\zeta \rangle \langle \kappa|
        \end{pmatrix} \right)
        \end{split} \\  
        = & \sum_{\zeta, \kappa=1}^{n} {\rm tr} \left(\begin{pmatrix}
            0 & 0 \\ |\kappa \rangle \langle \zeta| & 0
        \end{pmatrix} 
        \begin{pmatrix}
            0 & -|\zeta \rangle \langle \kappa|/n \\ -R(|\zeta \rangle \langle \kappa|)/n & 0
        \end{pmatrix} \right) 
        + \sum_{\zeta, \kappa=1}^{n} {\rm tr} \left(\begin{pmatrix}
            0 & |\kappa \rangle \langle \zeta| \\ 0 & 0
        \end{pmatrix} 
        \begin{pmatrix}
            0 & -R(|\zeta \rangle \langle \kappa|)/n \\ -|\zeta \rangle \langle \kappa|/n & 0
        \end{pmatrix} \right)\\  
        = & \sum_{\zeta, \kappa=1}^{n} {\rm tr} \left(\begin{pmatrix}
            0 & 0 \\ 0 & - |\kappa \rangle \langle \kappa |/n
        \end{pmatrix} \right) 
        + \sum_{\zeta, \kappa=1}^{n} {\rm tr} \left(\begin{pmatrix}
            - |\kappa \rangle \langle \kappa |/n & 0 \\ 0 & 0
        \end{pmatrix} \right) \\ 
        = & 2 \sum_{\zeta, \kappa=1}^{n} \frac{-1}{n} = -2 \frac{n^2}{n} = -2n
        \, .
    \end{align*}
    \end{widetext}

    In the third step we used that the first and the last sum vanish because $\Phi_{2n}$ ``swaps'' the two diagonal blocks meaning the product in the corresponding trace is always zero.
    Altogether, this shows that any map of the form~\eqref{eq:gen_Rob} fulfills condition (iv) from Theorem~\ref{thm_2}, showing that the corresponding entanglement witnesses are optimal.
\end{proof}
${}$
\begin{acknowledgments}

We would like to thank Jens Eisert for insightful discussions and for drawing our attention to key references on the separability problem.
We also thank Joonwoo Bae making us aware of the connection between our results and the SPA conjecture (cf.~Rem.~\ref{rem_kernel_weaker_than_spanning}~(ii)).
We are also grateful to Gregory White for valuable feedback on an earlier version of this manuscript.
FvE is funded by the \textit{Deutsche Forschungsgemeinschaft} (DFG, German Research Foundation) -- project number 384846402, and supported by the Einstein Foundation (Einstein Research Unit on Quantum Devices) and the MATH+ Cluster of Excellence. 
SC is funded by the Deutsche Forschungsgemeinschaft (DFG, German Research Foundation) under Germany´s Excellence Strategy – The Berlin Mathematics Research Center MATH+ (EXC-2046/1, project ID: 390685689).
\end{acknowledgments}

\bibliography{control21vJan20}

\begin{thebibliography}{70}
\bibitem [{foo()}]{footnote_fig1}%
  \BibitemOpen
  \href@noop {} {\bibinfo {title} {{The figure in question gets some details
  wrong---necessarily so, because three dimensions do not suffice to depict all
  topological properties of the separable states---such as the fact that the
  pure product states are path-connected (i.e., they do not split into two
  connected components). However, we found this to be the best visualization of
  (weak) optimality and the spanning property in terms of balancing accuracy
  and conveying the concepts in a simple enough manner.}}}\BibitemShut {Stop}%
\end{thebibliography}

\end{document}